\documentclass[final]{IEEEtran}
\usepackage{amsthm,amssymb,graphicx,multirow,amsmath,color,amsfonts,physics}
\usepackage[update,prepend]{epstopdf}
\usepackage[noadjust]{cite}
\usepackage{tikz}
\usepackage{bbm} 
\usepackage{pdfpages}
\usepackage{balance}
\usepackage{multirow}
\usepackage{comment}
\usepackage{subfigure}
\usepackage{yhmath}
\usepackage{amsmath}
\usepackage{stfloats}

\allowdisplaybreaks 
\allowdisplaybreaks 

\usepackage{float}

\begin{document}
\def\nba{{\mathbf{a}}}
\def\nbb{{\mathbf{b}}}
\def\nbc{{\mathbf{c}}}
\def\nbd{{\mathbf{d}}}
\def\nbe{{\mathbf{e}}}
\def\nbf{{\mathbf{f}}}
\def\nbg{{\mathbf{g}}}
\def\nbh{{\mathbf{h}}}
\def\nbi{{\mathbf{i}}}
\def\nbj{{\mathbf{j}}}
\def\nbk{{\mathbf{k}}}
\def\nbl{{\mathbf{l}}}
\def\nbm{{\mathbf{m}}}
\def\nbn{{\mathbf{n}}}
\def\nbo{{\mathbf{o}}}
\def\nbp{{\mathbf{p}}}
\def\nbq{{\mathbf{q}}}
\def\nbr{{\mathbf{r}}}
\def\nbs{{\mathbf{s}}}
\def\nbt{{\mathbf{t}}}
\def\nbu{{\mathbf{u}}}
\def\nbv{{\mathbf{v}}}
\def\nbw{{\mathbf{w}}}
\def\nbx{{\mathbf{x}}}
\def\nby{{\mathbf{y}}}
\def\nbz{{\mathbf{z}}}
\def\nb0{{\mathbf{0}}}
\def\nb1{{\mathbf{1}}}

\def\nbA{{\mathbf{A}}}
\def\nbB{{\mathbf{B}}}
\def\nbC{{\mathbf{C}}}
\def\nbD{{\mathbf{D}}}
\def\nbE{{\mathbf{E}}}
\def\nbF{{\mathbf{F}}}
\def\nbG{{\mathbf{G}}}
\def\nbH{{\mathbf{H}}}
\def\nbI{{\mathbf{I}}}
\def\nbJ{{\mathbf{J}}}
\def\nbK{{\mathbf{K}}}
\def\nbL{{\mathbf{L}}}
\def\nbM{{\mathbf{M}}}
\def\nbN{{\mathbf{N}}}
\def\nbO{{\mathbf{O}}}
\def\nbP{{\mathbf{P}}}
\def\nbQ{{\mathbf{Q}}}
\def\nbR{{\mathbf{R}}}
\def\nbS{{\mathbf{S}}}
\def\nbT{{\mathbf{T}}}
\def\nbU{{\mathbf{U}}}
\def\nbV{{\mathbf{V}}}
\def\nbW{{\mathbf{W}}}
\def\nbX{{\mathbf{X}}}
\def\nbY{{\mathbf{Y}}}
\def\nbZ{{\mathbf{Z}}}

\def\ncalA{{\mathcal{A}}}
\def\ncalB{{\mathcal{B}}}
\def\ncalC{{\mathcal{C}}}
\def\ncalD{{\mathcal{D}}}
\def\ncalE{{\mathcal{E}}}
\def\ncalF{{\mathcal{F}}}
\def\ncalG{{\mathcal{G}}}
\def\ncalH{{\mathcal{H}}}
\def\ncalI{{\mathcal{I}}}
\def\ncalJ{{\mathcal{J}}}
\def\ncalK{{\mathcal{K}}}
\def\ncalL{{\mathcal{L}}}
\def\ncalM{{\mathcal{M}}}
\def\ncalN{{\mathcal{N}}}
\def\ncalO{{\mathcal{O}}}
\def\ncalP{{\mathcal{P}}}
\def\ncalQ{{\mathcal{Q}}}
\def\ncalR{{\mathcal{R}}}
\def\ncalS{{\mathcal{S}}}
\def\ncalT{{\mathcal{T}}}
\def\ncalU{{\mathcal{U}}}
\def\ncalV{{\mathcal{V}}}
\def\ncalW{{\mathcal{W}}}
\def\ncalX{{\mathcal{X}}}
\def\ncalY{{\mathcal{Y}}}
\def\ncalZ{{\mathcal{Z}}}

\def\nbbA{{\mathbb{A}}}
\def\nbbB{{\mathbb{B}}}
\def\nbbC{{\mathbb{C}}}
\def\nbbD{{\mathbb{D}}}
\def\nbbE{{\mathbb{E}}}
\def\nbbF{{\mathbb{F}}}
\def\nbbG{{\mathbb{G}}}
\def\nbbH{{\mathbb{H}}}
\def\nbbI{{\mathbb{I}}}
\def\nbbJ{{\mathbb{J}}}
\def\nbbK{{\mathbb{K}}}
\def\nbbL{{\mathbb{L}}}
\def\nbbM{{\mathbb{M}}}
\def\nbbN{{\mathbb{N}}}
\def\nbbO{{\mathbb{O}}}
\def\nbbP{{\mathbb{P}}}
\def\nbbQ{{\mathbb{Q}}}
\def\nbbR{{\mathbb{R}}}
\def\nbbS{{\mathbb{S}}}
\def\nbbT{{\mathbb{T}}}
\def\nbbU{{\mathbb{U}}}
\def\nbbV{{\mathbb{V}}}
\def\nbbW{{\mathbb{W}}}
\def\nbbX{{\mathbb{X}}}
\def\nbbY{{\mathbb{Y}}}
\def\nbbZ{{\mathbb{Z}}}

\def\nfrakR{{\mathfrak{R}}}

\def\nrma{{\rm a}}
\def\nrmb{{\rm b}}
\def\nrmc{{\rm c}}
\def\nrmd{{\rm d}}
\def\nrme{{\rm e}}
\def\nrmf{{\rm f}}
\def\nrmg{{\rm g}}
\def\nrmh{{\rm h}}
\def\nrmi{{\rm i}}
\def\nrmj{{\rm j}}
\def\nrmk{{\rm k}}
\def\nrml{{\rm l}}
\def\nrmm{{\rm m}}
\def\nrmn{{\rm n}}
\def\nrmo{{\rm o}}
\def\nrmp{{\rm p}}
\def\nrmq{{\rm q}}
\def\nrmr{{\rm r}}
\def\nrms{{\rm s}}
\def\nrmt{{\rm t}}
\def\nrmu{{\rm u}}
\def\nrmv{{\rm v}}
\def\nrmw{{\rm w}}
\def\nrmx{{\rm x}}
\def\nrmy{{\rm y}}
\def\nrmz{{\rm z}}

\def\nbydef{:=}
\def\nborel{\ncalB(\nbbR)}
\def\nboreld{\ncalB(\nbbR^d)}
\def\sinc{{\rm sinc}}

\newtheorem{lemma}{Lemma}
\newtheorem{thm}{Theorem}
\newtheorem{definition}{Definition}
\newtheorem{ndef}{Definition}
\newtheorem{nrem}{Remark}
\newtheorem{theorem}{Theorem}
\newtheorem{prop}{Proposition}
\newtheorem{cor}{Corollary}
\newtheorem{example}{Example}
\newtheorem{remark}{Remark}
\newtheorem{assumption}{Assumption}
	

\newcommand{\ceil}[1]{\lceil #1\rceil}
\def\argmin{\operatorname{arg~min}}
\def\argmax{\operatorname{arg~max}}
\def\figref#1{Fig.\,\ref{#1}}%
\def\E{\mathbb{E}}
\def\EE{\mathbb{E}^{!o}}
\def\P{\mathbb{P}}
\def\pc{\mathtt{P_c}}
\def\rc{\mathtt{R_c}}   
\def\p{p}

\def\V{\operatorname{Var}}
\def\erfc{\operatorname{erfc}}
\def\erf{\operatorname{erf}}
\def\opt{\mathrm{opt}}
\def\R{\mathbb{R}}
\def\Z{\mathbb{Z}}

\def\LL{\mathcal{L}^{!o}}
\def\var{\operatorname{var}}
\def\supp{\operatorname{supp}}

\def\N{\sigma^2}
\def\T{\beta}							
\def\sinr{\mathtt{SINR}}			
\def\snr{\mathtt{SNR}}
\def\sir{\mathtt{SIR}}
\def\ase{\mathtt{ASE}}
\def\se{\mathtt{SE}}

\def\calN{\mathcal{N}}
\def\FE{\mathcal{F}}
\def\calA{\mathcal{A}}
\def\calK{\mathcal{K}}
\def\calT{\mathcal{T}}
\def\calB{\mathcal{B}}
\def\calE{\mathcal{E}}
\def\calP{\mathcal{P}}
\def\calL{\mathcal{L}}


\def\l{\ell}
\newcommand{\fad}[2]{\ensuremath{\mathtt{h}_{#1}[#2]}}
\newcommand{\h}[1]{\ensuremath{\mathtt{h}_{#1}}}

\newcommand{\err}[1]{\ensuremath{\operatorname{Err}(\eta,#1)}}
\newcommand{\FD}[1]{\ensuremath{|\mathcal{F}_{#1}|}}



\def\Bx{{\mathcal{B}}^x}
\def\Bxx{{\mathcal{B}}^{x_0}}
\def\jx{y}
\def\m{(\bar{n}-1)}
\def\mm{\bar{n}-1}
\def\Nx{{\mathcal{N}}^x}
\def\Nxo{{\mathcal{N}}^{x_0}}
\def\wj{w_{jx_0}}
\def\uij{u_{jx}}
 \def\yj{y}
 \def\yjx{y}
 \def\zjx{z_x}
 \def \tx {y_0}
 \def \htx {h_0}

\def\rx{z_{1}}
\def\ry{z_{2}}

\def\Rx{Z_{1}}
\def\Ry{Z_{2}}

\def \hyxx {h_{y_{x_0}}}
\def \hyx {h_{y_x}}

\def\nbb1{\mathbbm{1}}
\def\xi{\textbf{x}_i}
\def\xj{\textbf{x}_j}
\def\xk{\textbf{x}_k}
\def\xx{\textbf{x}_0}
\def\yk{\textbf{y}_k}
\def\yj{\textbf{y}_j}
\def\yy{\textbf{y}_0}
\def\oe{\textbf{o}_e}
\def\zl{\textbf{z}_l}
\def\wik{\textbf{w}_{i,k}}
\def\ie{{\em i.e. }}
\def\eg{{\em e.g. }}
\def\iid{{\em i.i.d. }}
\def\avg{\rm avg}

\def\rmnuma{\rm\uppercase\expandafter{\romannumeral1}}
\def\rmnumb{\rm\uppercase\expandafter{\romannumeral2}}
\def\rmnumc{\rm\uppercase\expandafter{\romannumeral3}}
\def\rmnumd{\rm\uppercase\expandafter{\romannumeral4}}
\def\rmnume{\rm\uppercase\expandafter{\romannumeral5}}
\def\rmnumf{\rm\uppercase\expandafter{\romannumeral6}}

\pagenumbering{gobble}
\graphicspath{{./Figures/}}
\title{HAPS-enabled Downlink Coverage Enhancement\\ in Islands and Maritime Areas}
\author{
 Hao Lin,~\IEEEmembership{Graduate Student Member,~IEEE},  Mustafa A. Kishk,~\IEEEmembership{Member,~IEEE}\\ and Mohamed-Slim Alouini,~\IEEEmembership{Fellow,~IEEE}
\thanks{Hao Lin is with the Electrical and Computer Engineering Program, Computer, Electrical and Mathematical Sciences and Engineering Division (CEMSE), King Abdullah University of Science and Technology (KAUST), Thuwal 23955-6900, Saudi Arabia (e-mail: hao.lin.std@gmail.com).\\
\indent Mustafa A. Kishk is with the Department of Electronic Engineering,
Maynooth University, Maynooth, W23 F2H6 Ireland (e-mail:
mustafa.kishk@mu.ie).\\
\indent Mohamed-Slim Alouini is with the CEMSE Division, King Abdullah
University of Science and Technology (KAUST), Thuwal 23955-6900,
Saudi Arabia (e-mail: slim.alouini@kaust.edu.sa).}
}
\maketitle
\vspace{-2cm}
\begin{abstract}
Non-terrestrial networks (NTNs) are poised to play a critical role in next-generation mobile communications, offering enhanced flexibility, improved line-of-sight (LoS) conditions, and overcoming the limitations of terrestrial networks (TNs). Among NTN platforms, high altitude platform stations (HAPSs) have emerged as a promising solution to provide Internet connectivity to underserved regions, including rural areas, islands, and maritime zones, where traditional infrastructure deployment is costly and challenging to deploy. In this paper, we investigate the feasibility of large-scale HAPS deployment to connect island and maritime users, considering real-world shadowing effects on part of HAPSs caused by the presence of island building clusters. We first analyze the coverage performance of onshore (island) and offshore (remote sea) users, in which the channels between HAPSs and the user follow the shadowed Rician distributions and Rician distributions, respectively. Next, we introduce an evaluation method for nearshore users in a hybrid channel environment with HAPSs, and propose approximations that can reduce computational complexity. Based on the simulation results, we discuss how the distance from the island boundary (i.e. the relative remoteness of maritime users) affects coverage performance under different HAPS densities. We also emphasize the importance of choosing a balanced HAPS density or an advanced HAPS deployment scheme.
\end{abstract}


\section{Introduction} \label{sec:Intro}
In the vision of next-generation mobile communications, the rapid adoption of advanced technologies such as remote education, smart healthcare, intelligent transportation, digital twins, and extended reality, has driven unprecedented demand for global seamless connectivity \cite{10628026,10070393,9369324,9349624}. Especially, with the development of marine resource utilization, ocean monitoring, and offshore facilities operation, maritime communications are becoming increasingly necessary \cite{fang2025study,10416859}. However, traditional terrestrial networks (TNs) face huge limitations in meeting these exploding demands due to high deployment costs, geographic challenges, and limited financial supports \cite{9042251}. One potential solution is to build wireless mesh networks among vessels, buoys or sensors \cite{10840318}, but this approach still faces issues such as interference management and capacity limitation. Long-Term Evolution (LTE) technology, supported by multiple-input multiple-output (MIMO) and carrier aggregation, can realize reliable direct ship-to-shore communications \cite{9939173,8694989}. Since the height of the onshore base stations is limited, it is still difficult to support stable coverage for wider ocean areas. Therefore, non-terrestrial networks (NTNs), including low-altitude platforms (LAPs), high-altitude platform stations (HAPSs), low-earth orbit (LEO) satellites, and geostationary orbit (GEO) satellites, have become a viable solution to extend connectivity to these areas \cite{10742580,10355086,10396843}. 

Different types of NTN platforms play their respective roles in future communication architectures due to their own characteristics, as shown in \cite{deng2025distributed}. For example, tethered unmanned aerial vehicles (UAVs) can help extend reliable communication to nearshore users \cite{10646299} and LEO satellites can be a solution to achieve Internet coverage for remote areas. HAPSs mainly operate at an altitude of 20-50 km, which can balance coverage range and propagation delay \cite{9380673}. Unlike LAPs, HAPSs can support larger communication payloads and advanced antenna arrays, enabling higher data rates and precise beamforming \cite{10355104}. 
Each HAPS has a wider coverage range of up to 140 km and an autonomous flight time of several months. Also, HAPSs have lower latency than LEO satellites because of their proximity to the earth's surface and provide more stable connections without orbital motion \cite{5995281}. Equipped with 5G and LTE access technologies, HAPSs can provide direct services to users without dedicated antennas. Although energy consumption remains a challenge, solutions such as solar panels, wind turbines, and tethering cables are being explored to ensure the sustainable operation of HAPSs \cite{10417095}.  Therefore, HAPSs are expected to play a pivotal role in bringing ubiquitous connectivity and supporting ultra-reliable low-latency communications (URLLC) to unconnected areas, such as islands and maritime areas \cite{koo2023simultaneous}. HAPSs can not only provide services to areas lacking TN infrastructure, but also fill the coverage gaps of existing networks in island areas, share the traffic demand from mobile users and massive machine-type communications, and provide backup optical feeder links. Integrated sensing and communications (ISAC) has been advanced to enhance the spectrum utilization and quality of services (QoS), particularly in intelligent transportation systems, low-altitude network management and massive IoT, where HAPSs can also offer a compelling solution for implementing ISAC systems for remote scenarios \cite{10417095,9380673,10926897}. Currently, Softbank Corp has planned to use HAPSs with a payload of up to 70 kg, to provide services in areas that are difficult to cover by the existing mobile networks \cite{Softbank}. To optimize the deployment of large-scale HAPS network, a performance analysis framework of HAPS networks in a complex electromagnetic environment is urgently needed.

Stochastic geometry has been widely used to study the coverage capabilities of terrestrial large-scale wireless networks (LSWN), without losing accuracy and tractability \cite{haenggi2012stochastic,7733098}. It is also important to characterize the performance of future large-scale non-terrestrial networks \cite{9378781}. System-level metrics including coverage probability, latency, capacity, and energy efficiency can be investigated using tools from stochastic geometry \cite{10634042}. Focusing on the HAPS-based solutions in maritime applications, we establish a unified mathematical framework to model the channels between HAPSs and onshore, nearshore, and offshore users. More details on related work are shown in Sec. \ref{sec:relatedwork} and the contributions of this paper are summarized in Sec. \ref{sec:contribution}.

\subsection{Related Work}\label{sec:relatedwork}

In this paper, we investigate the potential and performance of HAPS-based solutions for island and maritime areas. Therefore, we divide the related work into: (i) HAPS-based solutions for the unconnected and (ii) stochastic geometry for HAPS-based solutions.

\textit{HAPS-based solutions for the unconnected}: Connecting the unconnected is an important goal of future communication systems, and HAPSs play an important role in providing services to rural and remote areas that lack broadband connectivity and infrastructure \cite{arum2020review}. In \cite{giambene2022lora}, the authors investigated HAPS-based environmental monitoring applications for remote and unconnected regions, where ground Internet of Things (IoT) devices transmit information to HAPSs via long-range radio (LoRa) technique. They examined the impact of IoT device density, transmit power, and HAPS altitude on throughput and conditional success probability. Using numerical iterative methods, the authors in \cite{10379023} demonstrated the potential of heterogeneous networks to achieve fair access services in both urban and rural areas. They considered a system setup where a cloud-enabled HAPS can connect with high-altitude balloons and terrestrial base stations to serve both aerial and ground users. By jointly optimizing user association and beamforming, the authors in \cite{10304301} analyzed the performance of hybrid satellite-HAPS-terrestrial networks and discussed the prospects for connecting unconnected regions. They also explored the application of machine learning in such vertical heterogeneous networks (vHetNets) \cite{10239349}. Due to the large payload capacity and superior LoS conditions, HAPSs can be equipped with reconfigurable intelligent surfaces (RISs) and act as relay nodes, to further connect the unconnected base stations \cite{10184499}.\\
\indent \textit{Stochastic geometry for HAPS-based solutions}: Stochastic geometry has been widely applied to the performance analysis of large-scale wireless networks, including terrestrial, non-terrestrial, and integrated networks. By analyzing system-level metrics such as coverage probability, latency, and channel capacity, operators can optimize NTNs' performance and reduce capital expenditure (CAPEX) and operational expenditure (OPEX) \cite{10634042}.  For example, in \cite{10050345}, the authors proposed an efficient stochastic geometry framework to evaluate the quality of service for users inside and outside disaster-affected areas. This framework utilized low-altitude platforms (LAPs) and HAPSs to address disasters of varying scales, demonstrating the effectiveness of vertical heterogeneous networks in emergency scenarios.  HAPSs can not only serve as cellular or non-cellular base stations for mobile users, but also act as nodes for backhaul and computational networks \cite{10355104}. For instance, in \cite{9520123}, the authors investigated the performance of hybrid satellite-aerial-terrestrial networks, where HAPS serves as an aerial node with file-caching capabilities. Using stochastic geometry, they derived the outage probability and the hit probability of the considered vHetNet architecture. In \cite{10506977}, the authors explored a heterogeneous network based on unmanned aerial vehicles (UAVs), HAPSs, and low earth orbit (LEO) satellites, analyzing the contributions of each layer to overall performance. Under constraints of practical economic costs and signal-to-noise ratio (SNR), they proposed optimization schemes for total connection probability and discussed resource allocation strategies under different constraints. Furthermore, the future communication architecture should cover users and devices on the sea. In \cite{10040542}, the authors studied the coverage performance of space-air-ground-sea integrated networks (SAGSIN), where onshore base stations, tethered balloons, HAPSs, and satellites provide services to surface stations at sea. \\
\indent Unlike the above literature, our paper aims to build a mathematical framework for downlink performance evaluation considering the relationship between users, a shadowing zone and a homogeneous HAPS network, verify the feasibility of HAPS-based solution and provide guidance for future customized HAPS deployment to realize user fairness. More details about the contributions are provided in the next subsection.

\subsection{Contributions}\label{sec:contribution}
The contributions of this paper can be summarized as follows:
\begin{itemize}
    \item We develop a more accurate evaluation approach for HAPSs-enabled solutions on islands and maritime areas, considering the shadowing effect of the island building cluster on the channels between HAPSs and users. We propose to divide the HAPS deployment area into a shadowed Rician region and a Rician region according to the geometric relationship between user, island building cluster and HAPS network.
    \item We derive the exact expression of the coverage probability for onshore, nearshore, and offshore users using stochastic geometry. Especially for the nearshore users with a hybrid channel environment, we propose two approximations for coverage probability, whose computation complexity is less than the exact expression.
    \item We show that onshore users operate optimally with a high HAPS density, while offshore and nearshore users perform best with a low HAPS density. A balanced HAPS density can help minimize the coverage performance gap between users at different locations, and moreover, a customized deployment of the HAPS network has potential in the future.
\end{itemize}

The rest of this paper is organized as follows: In Sec. \ref{sec:SysMod}, we introduce the system model. In Sec. \ref{sec:coverageanalysis} we introduce the expression of coverage probabilities for the island (onshore) and remote maritime (offshore) users. In Sec. \ref{sec:nearshore}, we derive the exact expressions for nearshore users. We also propose two approximations for the coverage performance of nearshore users. In Sec. \ref{sec:simulation}, we conduct the Monte Carlo simulations and verify the exact expressions and approximations. Finally, we conclude this paper in Sec. \ref{sec:conclusion}.

\section{System Model} \label{sec:SysMod}
\begin{figure*}[htbp]
\centering
\subfigure[Case 1: Onshore users with HAPS-user links under shadowed Rician channels.]{
\begin{minipage}[b]{0.32\linewidth}
\centering 
\includegraphics[width=1\textwidth]{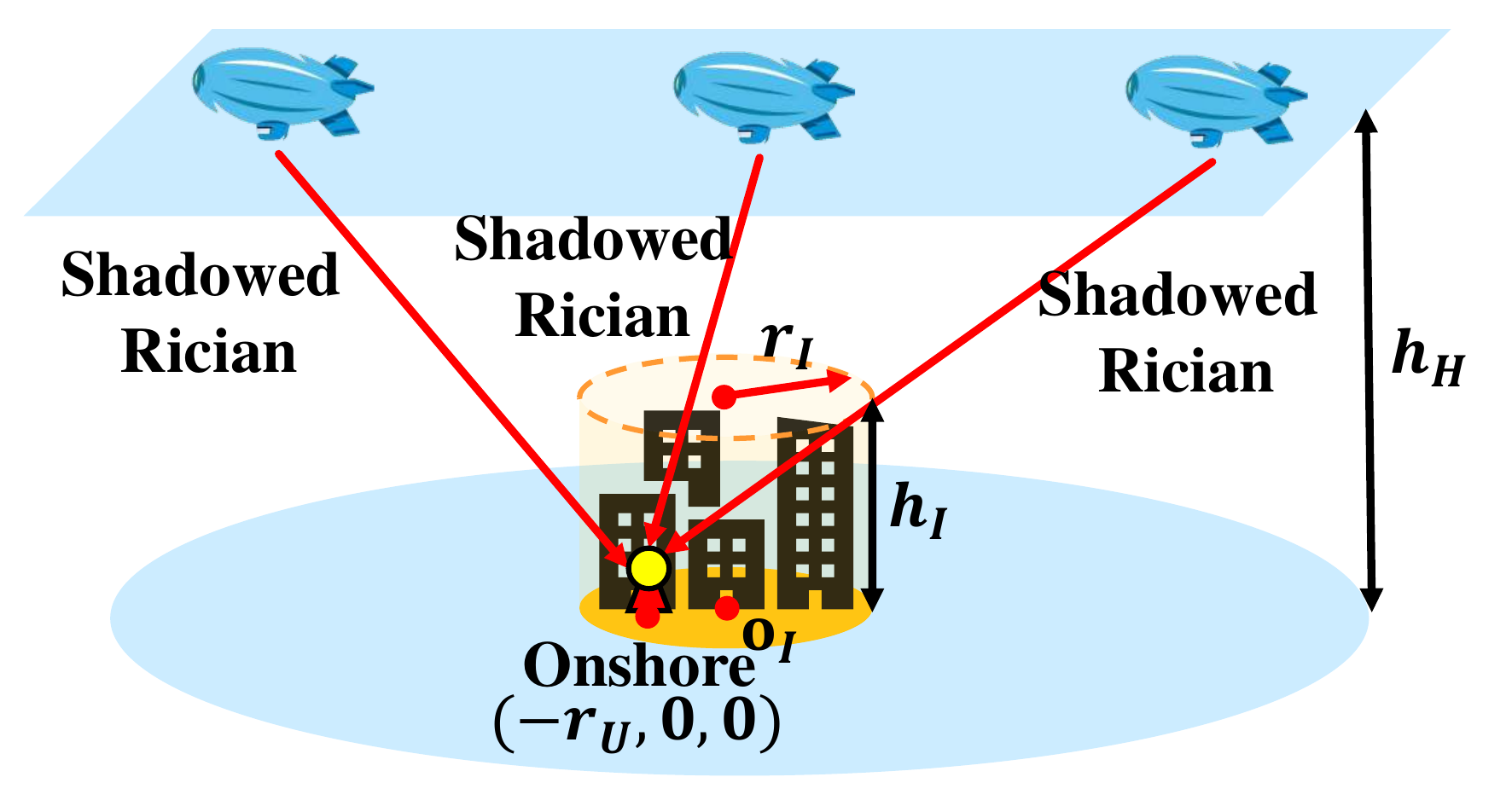}
\label{fig:Onshore}
\end{minipage}}
\hfill
\subfigure[Case 2: Nearshore users with HAPS-user links under shadowed Rician channels and Rician channels.]{
\begin{minipage}[b]{0.32\linewidth}
\centering 
\includegraphics[width=1\textwidth]{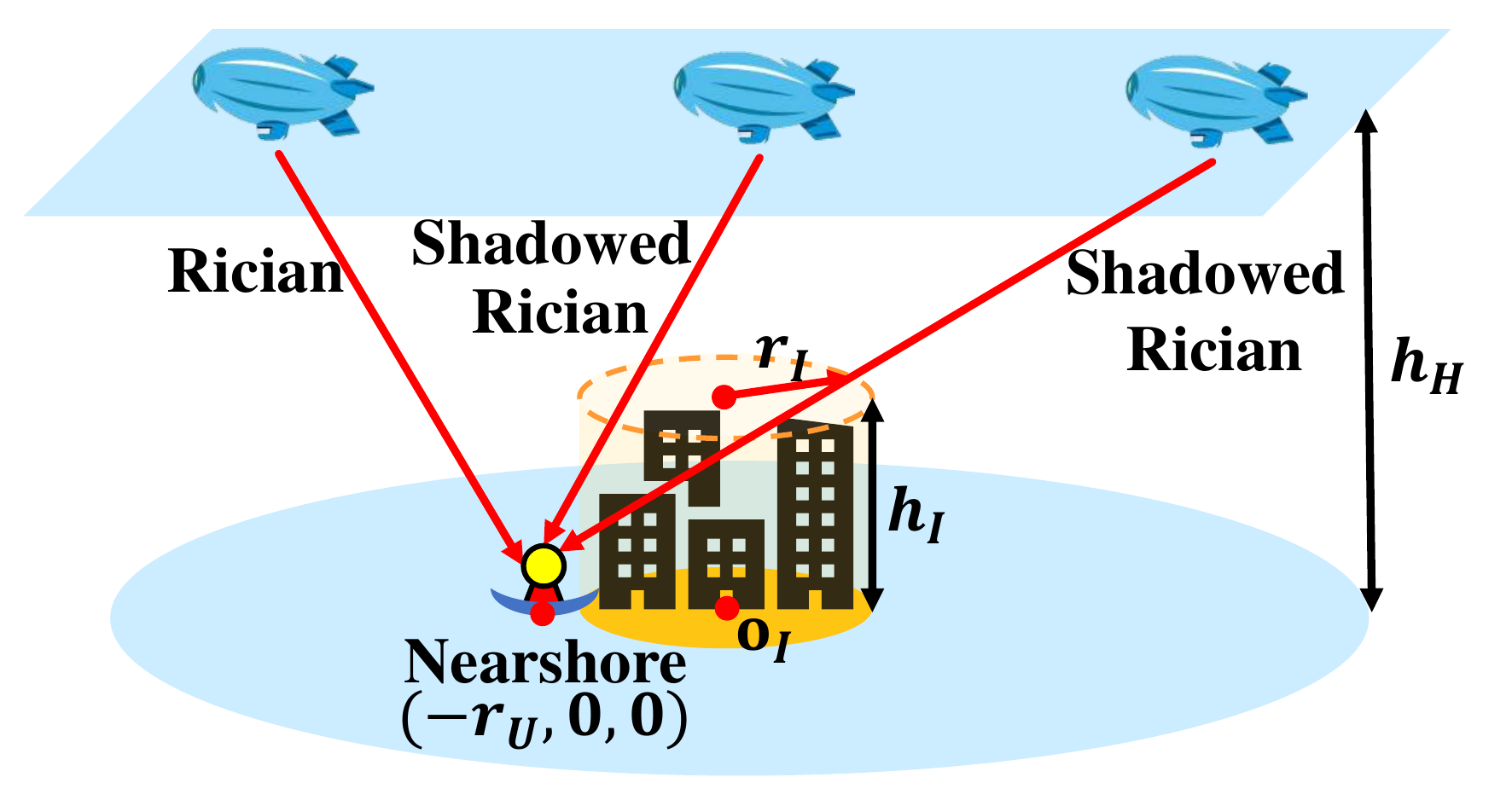}
\label{fig:Nearshore}
\end{minipage}}
\hfill
\subfigure[Case 3: Offshore users with HAPS-user links under Rician channels.]{
\begin{minipage}[b]{0.32\linewidth}
\centering 
\includegraphics[width=1\textwidth]{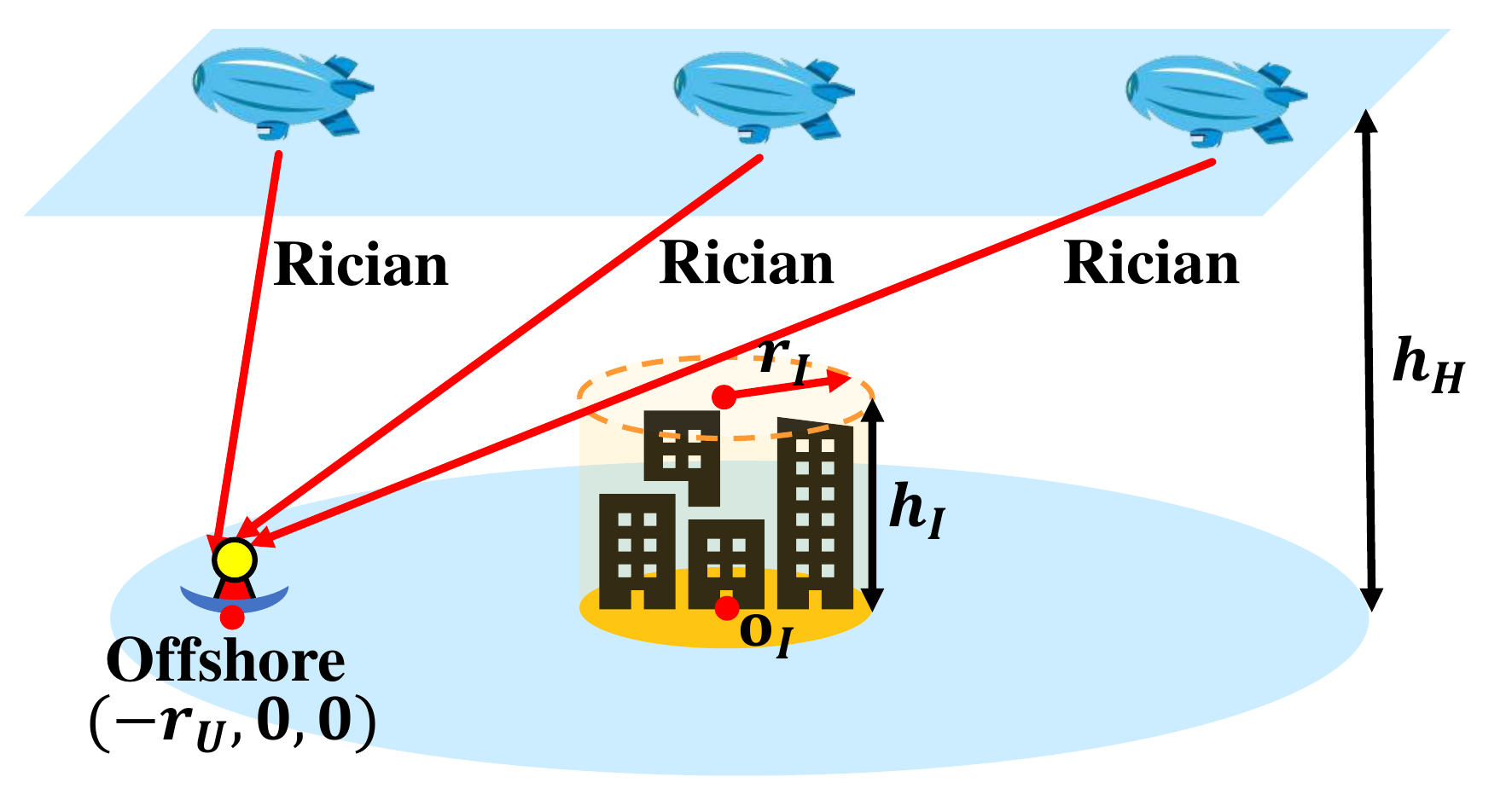}
\label{fig:Offshore}
\end{minipage}}
\caption{Illustration of channel environment for onshore, nearshore and offshore users, respectively.}
\label{fig:onnearoffshore}
\end{figure*}

In this paper, we use a circular area with radius $r_{I}$ to model a typical residential area (or forests) on an island, and the sea area is outside this circular area. We aim to investigate the downlink coverage performance of users at different geographic locations under the same large-scale HAPS deployment. Therefore, considering the randomness from the node failure and maintenance of the HAPS cluster, the locations of HAPSs follow a two-dimensional (2D) homogeneous Poisson point process (PPP) $\Phi=\{\textbf{x}_i\}$ with density $\lambda_{H}$ at altitude $h_{H}$. We assume that the center of island is $\textbf{o}_I=(0,0,0)$ and the location of a typical user is $\textbf{z}_U=(-r_U,0,0)$, where the distance between the typical user and the center of island is $r_{U}$. We also call $r_U-r_I$ the `relative distance to the island boundary', where a negative value corresponds to an onshore user and a positive value corresponds to a nearshore user or an offshore user.

Signal propagation is crucial for performance analysis of air-sea channels. A straightforward approach is to use the classic logarithmic path loss model and modify the parameters to account for the effects of obstacles and scatterers at sea. When the transmission antenna height is low (less than 1 km), the signal energy received by the user comes primarily from specular and diffuse reflections from the sea surface, rather than simply LoS path. However, when the transmission antenna height is higher, the majority of the signal energy received by the user comes from the LoS propagation path. This makes it particularly important to know whether the LoS signal is obstructed by random obstacles. Therefore, we model the cluster of buildings (or forests) on the island as a cylindrical area with the center of the bottom at $\textbf{o}_I$, radius $r_I$, and altitude $h_I$. When the LoS between a user and a HAPS crosses this cylindrical area, the power of the LoS component is randomly shadowed by the buildings, while the reflections on the buildings form the NLoS components of the signal. We assume that the shadowed LoS component follows a Gamma distribution and that the small-scale fading gain of the signal envelope is modeled using the shadowed Rician distribution. When the LoS between the user and a HAPS does not cross the cylindrical area, the small-scale fading gain of the signal envelope is modeled using the Rician distribution because clear LoS paths can be obtained, while reflections by the sea surface and evaporation ducting will form the NLoS component. In addition, as introduced in \cite{8528349,9939173}, the ducting effect leads to coverage extension when the electromagnetic waves are emitted in certain directions, the signals are concentrated in the intended direction, which is called beyond LoS (B-LoS) transmissions. Because HAPSs have broader footprints than the onshore BSs, we ignore the B-LoS effect but consider the multipath effect. Considering different locations of typical users, we discuss three cases:
\begin{itemize}
    \item \textit{Onshore case:} When $r_U$ is less than $r_I$, the typical user is on the island and almost all LoS paths between it and HAPSs are shadowed by the island building cluster. Therefore, the typical user is in a shadowed Rician channel environment where channel gains from all HAPSs to the typical user follow the shadowed Rician distribution.
    \item \textit{Nearshore case:} When $r_U$ is larger than $r_I$, the building cluster shadows the LoS paths between parts of the HAPSs and the typical user, and these HAPS operate in shadowed Rician fading channels. Since the building cluster does not shadow the LoS path between remaining HAPSs and the user, these HAPSs operate in Rician fading channels. Therefore, such a typical nearshore user operates in a hybrid channel environment.
    \item \textit{Offshore case:} When $r_U$ is much larger than $r_I$, the typical user is far from the island and can obtain clear LoS paths to almost all HAPSs. Therefore, the typical user is in a Rician channel environment, where the channel gains between all HAPSs and the typical user follow the Rician distribution.
\end{itemize}

It is worth noting that, when a typical user moves from the island center to the far sea area, its channel environment keeps changing. Specifically, as the user crosses the island boundary, over half of the HAPSs operate in the Rician channels because the building cluster no longer shadows LoS paths from the HAPSs on the sea side. When $r_U$ is larger than $r_I$ and continues to increase, more and more HAPSs can obtain clear LoS paths and operate in Rician channels. Therefore, the channel environment of the typical user gradually changes from a hybrid channel environment to a Rician channel environment. Furthermore, an increase in either the island's size or the height of buildings results in a greater number of HAPSs experiencing shadowing effects. For ease of reading, we summarize the main notation and parameters in Table \ref{tab:TableOfNotations}. We use $f(z)\big|_{b}^{a}$ to represent $f(a)-f(b)$.
\begin{table*}[htbp]
\caption{Table of Notations}
\centering
\begin{center}
\resizebox{\textwidth}{!}{
\renewcommand{\arraystretch}{1}
    \begin{tabular}{ {c} | {l} }
    \hline
        \hline
    \textbf{Notation} & \textbf{Description} \\ \hline
    $\Phi=\{\textbf{x}_i\}$; $\lambda_{H}$; $h_H$ & The set of locations of HAPSs; the density of HAPSs following a 2D PPP; the altitude of HAPSs  \\ \hline
    $\textbf{o}_I$; $r_I$; $h_I$ & The center of island; the radius of island; the altitude of buildings on the island \\ \hline
    $\textbf{z}_U$; $r_U$; $d_i$ & The location of the typical user; the distance between the typical user and island center; the distance between the typical user at $\textbf{z}_U$ and HAPS at $\textbf{x}_i$. \\ \hline
    $p_t$, $g_t$, $g_r$, $\lambda$ & Transmit power; transmitting gain; receiving gain; wave length of the carrier frequency \\ \hline
    $p_{I,0}$; $W_{I,i}$; $\alpha_{I}$ & The received power at a unit distance without channel gains; channel gain; path-loss exponent (for shadowed Rician channels);\\ \hline
    $b_I$; $m_I$; $\Omega_I$ & The standard deviation of NLoS components; the shape parameter; the average gain of the shadowed LoS component (for shadowed Rician channels) \\ \hline
    $\alpha_0$; $\beta_0$ & The parameters for approximation of shadowed Rician channel fading\\ \hline
    $P_{I,i}$& The signal power transmitted from HAPS at $\textbf{x}_i$ and received at the typical user via a shadowed Rician channel\\ \hline
    $p_{O,0}$; $W_{O,i}$; $\alpha_{O}$ & The received power at a unit distance without channel gains; channel gain; path-loss exponent (for Rician channels) \\ \hline
    $b_O$; $\Omega_O$ & The standard deviation of NLoS components; the average gain of the LoS component (for Rician channels) \\ \hline
    $P_{O,i}$& The signal power transmitted from HAPS at $\textbf{x}_i$ and received at the typical user via a Rician channel\\ \hline
    $s_I$; $s_O$& The atmospheric attenuation in shadowed Rician channels and Rician channels\\ \hline
    $\tau$; $N_0$; $P_{\rm cov}$& The decoding SINR threshold; the noise in the receiver circuit; coverage probability\\ \hline
     \hline
    \end{tabular}
}
\end{center}
\label{tab:TableOfNotations}
\end{table*}

\section{Downlink Analysis for Onshore and Offshore} \label{sec:coverageanalysis}
In this section, we analyze the downlink coverage performance for a typical user onshore or offshore, respectively. We specify the channel models between HAPSs and the user, and derive the expression of coverage probability for onshore users and offshore users, respectively.
\subsection{Onshore Users: Shadowed Rician Environment}
For users on the island, we have $r_U\leq r_I$. As shown in Fig. \ref{fig:Onshore}, the signals from HAPSs are not only affected by the atmosphere and reflections, but also shadowed by building cluster. Therefore, we assume that the channel gains of small-scale fading between HAPSs and the onshore user follow the shadowed Rician distributions \cite{10553646}. The power of signal transmitted from HAPS $i$ and received by the typical onshore user can be modeled as:
\begin{equation}
    P_{I,i}(d_i)=p_{t}(\frac{\lambda}{4\pi d_{I,i}})^{\alpha_I}g_tg_r s_I W_{I,i}=p_{I,0}W_{I,i}d_{I,i}^{-\alpha_I},
\label{PIidi}
\end{equation}
where
    $p_{I,0}=p_{t}(\frac{\lambda}{4\pi })^{\alpha_I}g_tg_r s_I$
is the received power at a unit distance regardless of the channel gains. $p_t$ is the transmitted power, $\lambda$ is the wave length, $g_t$, $g_r$ and $s_I$ represent the transmit antenna gain, receive antenna gain and the atmospheric attenuation in shadowed Rician channels, respectively. $\alpha_I$ is the path loss exponent for shadowed Rician channels and $d_{I,i}$ is the distance between the typical onshore user and HAPS located at $\textbf{x}_i$.
The shadowed Rician channel gain is represented using $W_{I,i}$, and the PDF of $W_{I,i}$ is:
\begin{equation}
\begin{array}{r@{}l}

    f_{W_{I,i}}(x)&= \frac{1}{2b_{I}}\big(\frac{2b_{I}m_I}{2b_{I}m_I+\Omega_{I}}\big)^{m_I}\exp \big(-\frac{1}{2b_{I}}x \big)\\
    & \times\;{}_{1}F_{1} \big(m_I;1;\frac{\Omega_{I}}{2b_{I}[2b_{I}m_{I}+\Omega_{I}]}x \big),
\end{array}
\label{fWIix}
\end{equation}
where $b_I$ is the standard deviation of NLoS component, $\Omega_{I}$ is the average of direct LoS component which follows the Gamma distribution, and $m_I$ is the shape parameter of the Gamma distribution. The CDF of $W_{I,i}$ can be calculated using:
\begin{equation}
\begin{array}{r@{}l}
    F_{W_{I,i}}(x)&= \big(\frac{2b_{I}m_{I}}{2b_{I}m_{I}+\Omega_{I}}\big)^{m_{I}}\\
    & \times\sum_{n=0}^{\infty}\frac{(m_{I})_n}{(1)_{n}n!} \big(\frac{\Omega_{I}}{2b_{I}m_{I}+\Omega_{I}} \big)^{n}\gamma\big(n+1,\frac{x}{2b_{I}}\big),
\end{array}
\label{FWIix}
\end{equation}
and the CCDF of $W_{I,i}$ is:
\begin{equation}
\begin{array}{r@{}l}
    \overline{F}&_{W_{I,i}}(x)\overset{\triangle}{=}1-F_{W_{I,i}}(x)= \big(\frac{2b_{I}m_{I}}{2b_{I}m_{I}+\Omega_{I}}\big)^{m_{I}}\\
    & \times\sum_{n=0}^{\infty}\frac{(m_{I})_n}{(1)_{n}n!} \big(\frac{\Omega_{I}}{2b_{I}m_{I}+\Omega_{I}} \big)^{n}\Gamma\big(n+1,\frac{x}{2b_{I}}\big),
\end{array}
\label{CCDFWIix}
\end{equation}
where $\gamma(s,x)=\int_0^{x}t^{s-1}e^{-t}\dd t$ and $\Gamma(s,x)=\int_x^{\infty}t^{s-1}e^{-t}\dd t$ represent the lower incomplete Gamma function and upper incomplete Gamma function with shape parameter $s$. The expectation of $W_{I,i}$ is $\mathbb{E}[W_{I,i}]=2b_I+\Omega_I$. Since the locations of HAPSs follow a 2D homogeneous PPP with density $\lambda_{H}$ and all HAPSs follow the same signal propagation model, the user selects the closest HAPS that has the strongest average received power. Therefore, we define the variable $D$ as the distance between the user and the closest HAPS and introduce the distribution of $D$ in Lemma \ref{lem:PDFofD}.
\begin{lemma}
    The PDF of the distance $d$ between the typical user and its nearest HAPS can be expressed as:
\begin{equation}
\begin{array}{r@{}}
    f_{D}(d)=2\lambda_{H}\pi\exp(\lambda_{H}\pi h_{H}^2) d\exp(-\lambda_{H}\pi d^2),
\end{array}
\label{PDFofD}
\end{equation}
for $h_H<d<\infty$.
\label{lem:PDFofD}
\end{lemma}
\begin{IEEEproof}
    See Appendix~\ref{app:PDFofD}.
\end{IEEEproof}

Therefore, the closest HAPS acts as the serving HAPS to the user, and other HAPSs act as interfering ones. Based on this, we introduce the coverage probability for typical onshore users in Theorem \ref{theo:onshore}.
\begin{theorem}
    When $r_U\leq r_I$, the channel gains between the typical user and HAPSs are modeled using shadowed Rician distributions, where the coverage probability of onshore users is:
\begin{equation}
\begin{array}{r@{}l}
    P_{\rm cov}(r_{U},\tau)&\overset{\triangle}{=}\P\{{\rm SINR}_{I}(r_U)>\tau\}\\
    &\approx \int_{h_H}^{\infty} \overline{F}_{W_{I,1}}\big(g_I(d)\big)f_{D}(d){\rm d}d,
\end{array}
\label{PcovOnshore}
\end{equation}
with
    $g_{I}(d)\approx \frac{\tau N_{0}}{p_{I,0}}d^{\alpha_{I}}+\frac{2\pi\lambda_{H}\tau (2b_{I}+\Omega_{I})d^{\alpha_{I}}}{\alpha_{I}-2}z^{2-\alpha_{I}} \big|_{\infty}^{d}$,
and $\overline{F}_{W_{I,1}}(x)$ is given in (\ref{CCDFWIix}).

\label{theo:onshore}
\end{theorem}
\begin{IEEEproof}
    See Appendix~\ref{app:onshore}.
\end{IEEEproof}
Theorem \ref{theo:onshore} describes the coverage probability of onshore users when the deployment range of HAPSs is infinite. Next, we adopt similar approaches to introduce the coverage probability of offshore users.

\subsection{Offshore Users: Rician Environment}
For the users in remote maritime areas, we have $r_U\gg r_I$. As shown in Fig. \ref{fig:Offshore}, the building cluster on the island does not shadow the LoS paths between the HAPSs and the user. The reflections on the sea surface and evaporation ducting lead to multiple NLoS components. Therefore, we assume that the channel gains of small-scale fading between HAPSs and offshore users follow a Rician distribution.
The power of signal transmitted from HAPS $i$ to the typical offshore user can be modeled as:
\begin{equation}
    P_{O,i}(d_i)=p_{t}(\frac{\lambda}{4\pi d_{O,i}})^{\alpha_O}g_tg_r s_O W_{O,i}=p_{O,0}W_{O,i}d_{O,i}^{-\alpha_O},
\label{POidi}
\end{equation}
where
    $p_{O,0}=p_{t}(\frac{\lambda}{4\pi })^{\alpha_O}g_tg_r s_O$
is the received power at a unit distance regardless of the channel gains. $s_O$ is the atmosphere attenuation in Rician channels, $\alpha_{O}$ is the path loss exponent for Rician channels, and $d_{O,i}$ is the distance between the typical offshore user and HAPS located at $\textbf{x}_i$.
The Rician channel gain is represented using $W_{O,i}$ and the PDF of $W_{O,i}$ is:
\begin{equation}
    f_{W_{O,i}}(x)=\frac{1}{2b_{O}}\exp \big(-\frac{x+\Omega_{O}}{2b_{O}} \big)I_{0} \big(\frac{\sqrt{\Omega_{O}x}}{b_{O}} \big),
\label{fWOix}
\end{equation}
where $b_O$ is the standard deviation of NLoS component, and $\Omega_{O}$ is the strength of direct LoS component. The CDF of $W_{O,i}$ is:
\begin{equation}
\begin{array}{r@{}l}
    F_{W_{O,i}}(x)& =1-Q_{1} \big(\sqrt{\frac{\Omega_{O}}{b_{O}}},\sqrt{\frac{x}{b_{O}}} \big),
\end{array}
\label{FWOix}
\end{equation}
and the CCDF of $W_{O,i}$ is
\begin{equation}
\begin{array}{r@{}l}
    \overline{F}_{W_{O,i}}(x)& =1-F_{W_{O,i}}(x)=Q_{1} \big(\sqrt{\frac{\Omega_{O}}{b_{O}}},\sqrt{\frac{x}{b_{O}}} \big),
\end{array}
\label{CCDFWOix}
\end{equation}
where
    $Q_1(a,b)=\int_{b}^{\infty}t\exp(-\frac{t^2+a^2}{2})I_{0}(at){\rm d}t$
is the Marcum Q-function of order $\nu=1$ and $I_0$ is the modified Bessel function of first kind of order $\nu-1=0$. The expectation of $W_{O,i}$ is $\mathbb{E}[W_{O,i}]=2b_O+\Omega_O$. Because all HAPSs operate in the Rician channels, the closest HAPS can provide the strongest average power, similar to the onshore case.  Therefore, the closest HAPS acts as the serving HAPS to the user and other HAPSs act as interfering ones. Based on this, we introduce the coverage probability for typical offshore users in Theorem \ref{theo:offshore}.

\begin{theorem}
When $r_U \gg r_I$, the channel gains between the typical user and HAPSs can be modeled using Rician distributions, where the coverage probability is:
\begin{equation}
\begin{array}{r@{}l}
    P_{\rm cov}(r_U,\tau)&\overset{\triangle}{=}\P\{{\rm SINR}_{O}(r_U)>\tau\}\\
    &\approx \int_{h_H}^{\infty} \overline{F}_{W_{O,i}}\big(g_{O}(d)\big)f_{D}(d){\rm d}d,
\end{array}
\label{PcovOffshore}
\end{equation}
with
    $g_{O}(d)\approx \frac{\tau N_{0}}{p_{O,0}}d^{\alpha_{O}}+\frac{2\pi\lambda_{H}\tau (2b_{O}+\Omega_{O})d^{\alpha_{O}}}{\alpha_{O}-2}z^{2-\alpha_{O}} \big|_{\infty}^{d}$,
and $\overline{F}_{W_{O,1}}(x)$ is given in (\ref{CCDFWOix}).
\label{theo:offshore}
\end{theorem}

\begin{IEEEproof}
    See Appendix~\ref{app:offshore}.
\end{IEEEproof}

However, for nearshore users, the channel environment is more complex than in the onshore and offshore cases. Therefore, in the next section, we introduce the coverage analysis for nearshore users.

\section{Downlink Analysis for Nearshore}\label{sec:nearshore}
Different from the onshore and offshore cases, nearshore users receive the signals from HAPSs with different LoS conditions. In this section, we propose to divide the deployment area of HAPSs into a shadowed Rician region and a Rician region to evaluate the coverage probability and the approximations.
\subsection{Nearshore Users: Hybrid Channel Environment}
For the nearshore users, we have $r_U>r_I$ but $r_U$ is not much larger than $r_I$. As shown in Fig. \ref{fig:Nearshore}, when the LoS path between one HAPS and the typical user is shadowed by the cylindrical building cluster, we assume that the channel gain follows the shadowed Rician distribution in (\ref{fWIix}) and (\ref{FWIix}), and the signal propagation model is the same as in (\ref{PIidi}). Otherwise, the channel gain follows the Rician distribution in (\ref{fWOix}) and (\ref{FWOix}), and the signal propagation model is the same as in (\ref{POidi}). The comparison between the channel environments of onshore, nearshore and offshore user is shown in Fig. \ref{fig:onnearoffshore}. 

We show the 3D illustration and the 2D illustration for different channel regions of HAPSs in Fig. \ref{fig:3DNearshore} and Fig. \ref{fig:Accurate}, considering a typical nearshore user located at $\textbf{x}_U=(-r_U,0,0)$. We project the top of the island at $(-r_U,0,h_I)$ from the typical user to the HAPS plane at altitude $h_H$, which is also a circular area centered at $\textbf{o}_I'=(r_U(\frac{h_H-h_I}{h_I}),0,h_H)$ with radius $r_U\frac{h_H}{h_I}$:
    $\mathcal{A}= \big\{\textbf{y}_l:\|\textbf{y}_l-\textbf{o}_I'\|\leq r_U\frac{h_H}{h_I} \big\}$.
We consider $\textbf{z}_U'=(-r_U,0,h_H)$ on the HAPS plane which is directly above the typical user $\textbf{z}_U$. From the perspective of $\textbf{z}_U'$, the area $\mathcal{A}$ and the areas blocked by $\mathcal{A}$ form the shadowed Rician region, while other areas form the Rician region. 

\begin{figure}[ht]
    \centering
    \includegraphics[width=0.85\linewidth]{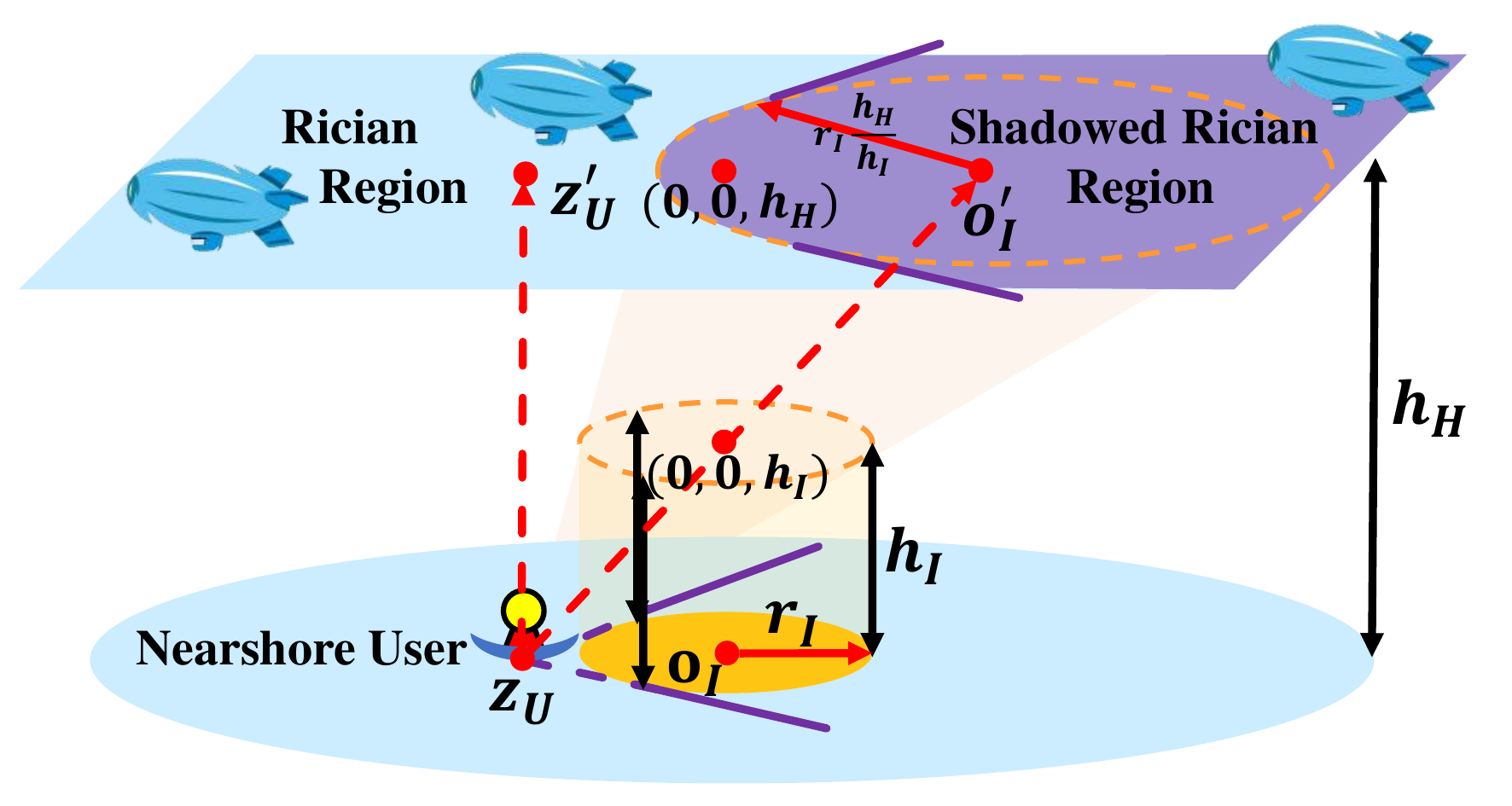}
    \caption{3D Illustration of hybrid channel environment for typical nearshore user.}
    \label{fig:3DNearshore}
\end{figure}

\begin{figure}[ht]
    \centering
    \includegraphics[width=0.85\linewidth]{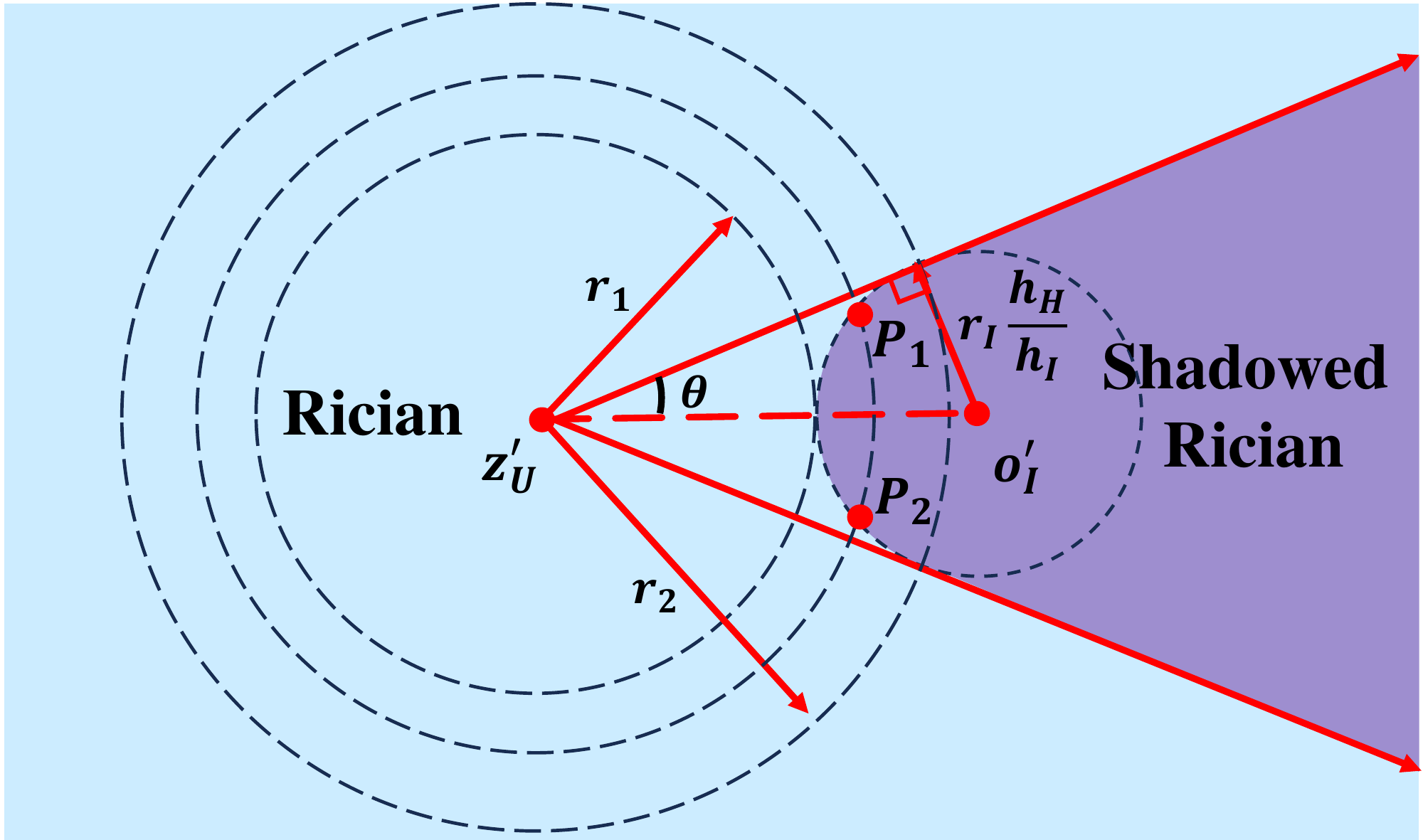}
    \caption{2D Illustration of hybrid channel environment for typical nearshore user.}
    \label{fig:Accurate}
\end{figure}

To calculate the coverage probability, we assume that the locations of the HAPSs operating in the Rician channels follow a PPP with density $\lambda_{H}$ and the locations of the HAPSs operating in the shadowed Rician channels follow another PPP with density $\lambda_{H}$. Because the closest HAPS does not always provide the highest signal power in the hybrid environment, we realize the selection strategy in two steps: (i) select the closest HAPS in the Rician channel region and the closest HAPS in the shadowed Rician channel region, and (ii) compare their average power and select the stronger one. We assume that the closest HAPS with a Rician channel is located at $\textbf{x}_{1}^{O}$, whose distance to the typical user is $d_{O,1}$. Similarly, the closest HAPS with a shadowed Rician channel is located at $\textbf{x}_{1}^{I}$, where the distance to the typical user is $d_{I,1}$. Therefore, we construct below two transmitter (TX) selection cases:
    $\textbf{H}_{O}\overset{\triangle}{=}\{{\rm Select\;the\;HAPS\;located\;at}\;\textbf{x}_{1}^{O}\}$,
and
    $\textbf{H}_{I}\overset{\triangle}{=}\{{\rm Select\;the\;HAPS\;located\;at}\;\textbf{x}_{1}^{I}\}$.
Because the TX selection depends on the average power of these two closest HAPSs, we formulate the TX selection as: 
\begin{equation}
    p_{O,0}\mathbb{E}[W_{O,1}] d_{O,1}^{-\alpha_O}\underset{\textbf{H}_I}{\overset{\textbf{H}_O}{\gtrless}} p_{I,0}\mathbb{E}[W_{I,1}] d_{I,1}^{-\alpha_I}.
\label{hypothesis}
\end{equation}

We define two functions:
\begin{equation}
\begin{array}{r@{}l}
\mathcal{T}_{I}(d_O)
& = \big(\frac{p_{I,0}(\Omega_I+2b_I)}{p_{O,0}(\Omega_O+2b_O)} \big)^{\frac{1}{\alpha_I}} d_{O}^{\frac{\alpha_O}{\alpha_I}},
\end{array}
\end{equation}
and 
\begin{equation}
\begin{array}{r@{}l}
\mathcal{T}_{O}(d_I)
& = \big(\frac{p_{O,0}(\Omega_O+2b_O)}{p_{I,0}(\Omega_I+2b_I)} \big)^{\frac{1}{\alpha_O}} d_{I}^{\frac{\alpha_I}{\alpha_O}},
\end{array}
\end{equation}
so the two TX selection cases can be rewritten using:
\begin{equation}
     d_{O,1}\underset{\textbf{H}_O}{\overset{\textbf{H}_I}{\gtrless}} \mathcal{T}_{O}(d_{I,1})
{\;\rm or\;}
     d_{I,1}\underset{\textbf{H}_I}{\overset{\textbf{H}_O}{\gtrless}} \mathcal{T}_{I}(d_{O,1}),
\end{equation}
and the coverage probability can be calculated using:
\begin{equation}
    P_{\rm cov}(r_U,\tau)=P_{\rm cov}(r_U,\tau,\textbf{H}_O)+P_{\rm cov}(r_U,\tau,\textbf{H}_I).
\label{PcovSum}
\end{equation}

Especially, if $p_{I,0}\mathbb{E}[W_{I,1}]=p_{O,0}\mathbb{E}[W_{O,1}]$ and $\alpha_I=\alpha_O$, the selection strategy is reduced to the closest TX selection strategy. When $\mathbf{H}_O$ is always satisfied, the typical user is in the Rician channel environment, which is the offshore case.

\subsection{Coverage Probability for Nearshore}
To obtain the exact expression of the coverage probability for nearshore users, we first derive the distribution of the distance $d_{O,1}$ between the typical user $\textbf{z}_{U}$ and its closest HAPS in the Rician region located at $\textbf{x}_{1}^{O}$. We analyze the characteristics of the Rician region:\\
\textbf{Rician Region:} As shown in Fig. \ref{fig:Accurate}, from the perspective of $\textbf{z}_U'$, the Rician region can be divided into: (i) a sector area centered at $\textbf{z}_U'$, with infinite radius and central angle $2\pi-2\theta$ where $\theta=\arcsin\frac{r_I}{r_U}$, and (ii) a sector area centered at $\textbf{z}_U'$, with radius $r_1=\frac{h_H}{h_I}(r_U-r_I)$ and central angle $2\theta$, except for the regions in $\mathcal{A}$. Therefore, as the search radius increases from $r_1$, the range of directional angle with LoS decreases, until the search radius is greater than $r_2=\frac{h_H}{h_I}\sqrt{r_U^2-r_I^2}$. From the perspective of the typical user located at $\textbf{z}_{U}$, the critical values of $d_{O,1}$ are $d_{th,1}=\sqrt{r_1^2+h_H^2}$ and $d_{th,2}=\sqrt{r_2^2+h_H^2}$.

Based on this, we can obtain the CDF and PDF of $d_{O,1}$.

\begin{lemma}
    Assume that $d_{O,1}$ is the distance between the typical nearshore user located at $\textbf{z}_U$ and its closest HAPS in the Rician region located at $\textbf{x}_1^{O}$. The CDF of $d_{O,1}$ is shown in (\ref{CDFofdI}) on the top of the next page, and the PDF of $d_{O,1}$ is
\begin{equation}
\begin{array}{l@{}l}
&f(d_{O,1})=2\lambda_{H}d_{O,1}(1-F(d_{O,1}))\\
&\times (\pi-\phi(d_{O,1})\mathbbm{1}(d_{th,1}<d_{O,1}<d_{th,2})-\theta\mathbbm{1}(d_{O,1}>d_{th,2})),
\end{array}
\end{equation}
where
\begin{equation}
\begin{array}{r@{}l}
    \phi&(d)= \arccos \big(\frac{h_I^2 d^2-h_I^2 h_H^2+r_U^2 h_H^2-r_I^2 h_H^2}{2\sqrt{d^2-h_H^2} r_U h_H h_I} \big),
\end{array}
\label{phi}
\end{equation}
and 
\begin{equation}
\begin{array}{r@{}l}
    \eta&(d)= \arccos \big(\frac{r_U^2h_H^2+r_I^2h_H^2-d^2h_I^2+h_H^2h_I^2}{2r_I r_Uh_H^2} \big).
\end{array}
\label{eta}
\end{equation}
\label{lem:PDFofdO1}
\end{lemma}
\begin{IEEEproof}
    See Appendix \ref{app:PDFofdO1}.
\end{IEEEproof}

\begin{figure*}[t]
\begin{equation}
F(d_{O,1})=\left\{
\begin{array}{l@{}l}
1-\exp(-\lambda_{H} \pi d_{O,1}^2+\lambda_{H}\pi h_H^2)&,h_H\leq d_{O,1}\leq d_{th,1}\\
1-\exp(-\lambda_{H}(\pi-\phi(d_{O,1}))(d_{O,1}^2-h_H^2)-\lambda_{H}\sin\eta(d_{O,1})r_Ir_U\frac{h_H^2}{h_I^2}+\lambda_{H}\eta(d_{O,1})r_I^2\frac{h_H^2}{h_I^2})&,d_{th,1}<d_{O,1}<d_{th,2}\\
1-\exp(-\lambda_{H}(\pi-\theta)(d_{O,1}^2-h_H^2)-\lambda_{H}\cot\theta r_I^2\frac{h_H^2}{h_I^2}+\lambda_{H}(\frac{\pi}{2}-\theta)r_I^2\frac{h_H^2}{h_I^2})&,d_{O,1}\geq d_{th,2}
\end{array}
\right.
\label{CDFofdI}
\end{equation}
\begin{equation}
F(d_{I,1})=\left\{
\begin{array}{l@{}l}
1-\exp(-\lambda_{H}\phi(d_{O,1})(d_{I,1}^2-h_H^2)+\lambda_{H}\sin\eta(d_{I,1})r_Ir_U\frac{h_H^2}{h_I^2}-\lambda_{H}\eta(d_{I,1})r_I^2\frac{h_H^2}{h_I^2})&,d_{th,1}\leq d_{I,1}<d_{th,2}\\
1-\exp(-\lambda_{H}\theta(d_{I,1}^2-h_H^2)+\lambda_{H}\cot\theta r_I^2\frac{h_H^2}{h_I^2}-\lambda_{H}(\frac{\pi}{2}-\theta)r_I^2\frac{h_H^2}{h_I^2})&,d_{I,1}\geq d_{th,2}
\end{array}
\right.
\label{CDFofdO}
\end{equation}
\end{figure*}
Next, we analyze the characteristics of the shadowed Rician region:\\
\textbf{Shadowed Rician Region:} As shown in Fig. \ref{fig:Accurate}, from the perspective of $\textbf{z}_U'$, the shadowed Rician region can be divided into: (i) the intersection between $\mathcal{A}$ and the sector area centered at $\textbf{z}_U'$, with radius $r_2=\frac{h_H}{h_I}\sqrt{r_U^2-r_I^2}$ and central angle $2\theta$, and (ii) an annular sector area centered at $\textbf{z}_U'$, with inner radius $r_2$, infinite outer radius, and central angle $2\theta$. From the perspective of the typical user located at $\textbf{z}_U$, the range of $d_{I,1}$ is from $d_{th,1}=\sqrt{r_1^2+h_H^2}$ to infinity, where the critical value is $d_{th,2}=\sqrt{r_2^2+h_H^2}$.

Based on this, we can obtain the CDF and PDF of $d_{I,1}$.

\begin{lemma}
    Assume that $d_{I,1}$ is the distance between the typical nearshore user located at $\textbf{z}_U$ and its closest HAPS in the shadowed Rician region located at $\textbf{x}_{1}^{I}$. The CDF of $d_{I,1}$ is shown in (\ref{CDFofdO}) on the top of the next page, and the PDF of $d_{I,1}$ is:
\begin{equation}
\begin{array}{l@{}l}
&f(d_{I,1})=2\lambda_{H}d_{I,1}(1-F(d_{I,1}))\\
&\times (\phi(d_{I,1})\mathbbm{1}(d_{th,1}<d_{I,1}<d_{th,2})+\theta\mathbbm{1}(d_{I,1}>d_{th,2})),
\end{array}
\end{equation}
where $\phi(d_{I,1})$ and $\eta(d_{I,1})$ have been shown in (\ref{phi}) and (\ref{eta}).
\label{lem:PDFofdI1}
\end{lemma}
\begin{IEEEproof}
    See Appendix \ref{app:PDFofdI1}.
\end{IEEEproof}

Because HAPSs are separated into two independent PPPs, $d_{O,1}$ and $d_{I,1}$ are independent, and the joint PDF can be written as the multiplication of PDFs of $d_{O,1}$ and $d_{I,1}$:
\begin{equation}
    f(d_{O,1},d_{I,1})=f(d_{O,1})f(d_{I,1}),
\end{equation}
and we can derive the coverage probability:
\begin{theorem}
    The coverage probability of nearshore users can be formulated using
\begin{equation}
    P_{\rm cov}(r_U,\tau)=P_{\rm cov}(r_U,\tau,\textbf{H}_O)+P_{\rm cov}(r_U,\tau,\textbf{H}_I).
\end{equation}

The coverage probability of nearshore users with $\textbf{H}_O$ is:
\begin{equation}
\begin{array}{r@{}l}
    P&_{\rm cov}(r_U,\tau,\textbf{H}_O)= \int_{h_H}^{\infty} \int_{\max\{d_{th,1},\mathcal{T}_{I}(d_{O,1})\}}^{\infty}\\
    &\overline{F}_{W_{O,1}}\big(g_{O}^{h}(d_{O,1},d_{I,1})\big) f(d_{O,1})f(d_{I,1})\;\dd d_{I,1}\dd d_{O,1},
\end{array}
\end{equation}
with
    $g_{O}^{h}(d_{O,1},d_{I,1})= \frac{\tau}{p_{O,0}}d_{O,1}^{\alpha_{O}} \big[N_{0}+I(d_{O,1},d_{I,1}) \big]$.

The coverage probability of nearshore users with $\textbf{H}_I$ is:
\begin{equation}
\begin{array}{r@{}l}
    P&_{\rm cov}(r_U,\tau,\textbf{H}_I)= \int_{d_{th,1}}^{\infty}\int_{\max\{h_H,\mathcal{T}_{O}(d_{I,1})\}}^{\infty}\\ &\overline{F}_{W_{I,1}}\big(g_I^{h}(d_{O,1},d_{I,1})\big)f(d_{O,1})f(d_{I,1})\;\dd d_{O,1}\dd d_{I,1},
\end{array}
\end{equation}
with 
    $g_{I}^{h}(d_{O,1},d_{I,1})= \frac{\tau}{p_{I,0}}d_{I,1}^{\alpha_{I}} \big[N_{0}+I(d_{O,1},d_{I,1}) \big]$.

The interference is approximated using:
\begin{equation}
\begin{array}{r@{}l}
I&(d_{O,1},d_{I,1})\approx   p_{O,0}(2b_{O}+\Omega_{O})\frac{(2\pi-2\theta)\lambda_{H}}{\alpha_O-2}z^{2-\alpha_O} \big|_{\infty}^{d_{O,1}}\\
& +p_{O,0}(2b_{O}+\Omega_{O})\frac{\lambda_{H}}{\alpha_O-2}\int_{-\theta}^{\theta}z^{2-\alpha_O} \big|_{\mathcal{T}_d(t)}^{\min\{d_{O,1},\mathcal{T}_d(t)\}} \dd t\\
& +p_{I,0}(2b_{I}+\Omega_{I})\frac{\lambda_{H}}{\alpha_I-2}\int_{-\theta}^{\theta}z^{2-\alpha_I} \big|_{\infty}^{\max\{d_{I,1},\mathcal{T}_d(t)\}} \dd t,
\end{array}
\end{equation}
where
    $\mathcal{T}_d(t)=\frac{h_H}{h_I}(r_U\cos{t}-\sqrt{r_I^2-r_U^2\sin^2{t}})$.
\label{theo:accurate_for_SA}
\end{theorem}
\begin{IEEEproof}
    See Appendix~\ref{app:accurate_for_SA}.
\end{IEEEproof}
However, the computational resource requirement for this exact expression of the coverage probability is high due to the triple integral of the angle $t$, $d_{O,1}$ and $d_{I,1}$. Therefore, we propose an approximation method by changing the geometric boundary of the Rician region and the shadowed Rician region.

\subsection{Overestimating the Shadowed Rician region}
To simplify the calculation of coverage probability for nearshore users, we propose an approximation by overestimating the shadowed Rician region.\\
\textbf{Approximation Method 1:} As shown in Fig. \ref{fig:Approx1}, we overestimate the area of the shadowed Rician region. From the perspective of $\textbf{z}_U'$, the shadowed Rician region is an annular sector area centered at $\textbf{z}_U'$, with inner radius $r_1=\frac{h_H}{h_I}(r_U-r_I)$, infinite outer radius and central angle $2\theta$. Therefore, the Rician region can be divided into: (i) a sector area centered at $\textbf{z}_U'$, with infinite radius and central angle $2\pi-2\theta$, and (ii) a sector area centered at $\textbf{z}_U'$, with radius $r_1$ and central angle $2\theta$. From the perspective of $\textbf{z}_U$, the critical value of the distance to HAPSs is $d_{th,1}=\sqrt{r_1^2+h_H^2}$.

Based on this, we can separate the set of HAPSs into two independent PPPs. The approximate CDFs and PDFs of $d_{O,1}$ and $d_{I,1}$ are formulated in Lemma \ref{lem:CDFPDFofdO1dI1}.
\begin{lemma}
    Based on the approximation method 1, the CDF of the distance $d_{O,1}$ between the typical nearshore user and its closest HAPS in the Rician region is:
\begin{equation}
\begin{array}{r@{}l}
    F&(d_{O,1})\\
    &=1-\exp(-\lambda_H (\pi-\theta)d_{O,1}^2-\lambda_{H}\theta\min(d_{O,1},d_{th,1})^2)\\
    &\times \exp(\lambda_H \pi h_H^2),
\end{array}
\end{equation}
and the PDF of $d_{O,1}$ is:
\begin{equation}
\begin{array}{r@{}l}
f&(d_{O,1})=2\lambda_H d_{O,1} \big((\pi-\theta)+\theta\mathbbm{1}(d_{O,1}<d_{th,1}) \big)\\
&\times\big(1-F(d_{O,1})\big),
\end{array}
\end{equation}
for $h_H<d_{O,1}<\infty$.

The CDF of the distance $d_{I,1}$ between the typical nearshore user and its closest HAPS in the shadowed Rician region is:
\begin{equation}
\begin{array}{r@{}l}
F(d_{I,1})=1-\exp(-\lambda_{H}\theta(d_{I,1}^2-d_{th,1}^2)),
\end{array}
\end{equation}
and the PDF of $d_{I,1}$ is
\begin{equation}
\begin{array}{r@{}l}
f(d_{I,1})&=2\lambda_{H}\theta d_{I,1}(1-F(d_{I,1})),
\end{array}
\end{equation}
for $d_{th,1}<d_{I,1}<\infty$.
\label{lem:CDFPDFofdO1dI1}
\end{lemma}
\begin{IEEEproof}
    See Appendix~\ref{app:CDFPDFofdO1dI1}.
\end{IEEEproof}

Based on these, we formulate the approximation of coverage probability in Theorem \ref{theo:approx1_for_SA}.

\begin{theorem}
    By overestimating the shadowed Rician area, the coverage probability for nearshore users can be represented using the sum of the coverage probabilities with $\textbf{H}_O$ and $\textbf{H}_I$, that is:
\begin{equation}
    P_{\rm cov}(r_U,\tau)=P_{\rm cov}(r_U,\tau,\textbf{H}_O)+P_{\rm cov}(r_U,\tau,\textbf{H}_I),
\end{equation}
where $d_{O,1}>h_H$ and $d_{I,1}>d_{th,1}$.

The coverage probability of nearshore users with $\textbf{H}_O$ is
\begin{equation}
\begin{array}{r@{}l}
    P&_{\rm cov}(r_U,\tau,\textbf{H}_O)= \int_{h_H}^{\infty} \int_{\max\{d_{th,1},\mathcal{T}_{I}(d_{O,1})\}}^{\infty}\\ &\overline{F}_{W_{O,1}}\big(g_{O}^{h}(d_{O,1},d_{I,1})\big) f(d_{O,1})f(d_{I,1})\;\dd d_{I,1}\dd d_{O,1},
\end{array}
\end{equation}
with
    $g_{O}^{h}(d_{O,1},d_{I,1})= \frac{\tau}{p_{O,0}}d_{O,1}^{\alpha_{O}} \big[N_{0}+I(d_{O,1},d_{I,1}) \big]$.

The coverage probability of nearshore users with $\textbf{H}_I$ is
\begin{equation}
\begin{array}{r@{}l}
    P&_{\rm cov}(r_U,\tau,\textbf{H}_I)= \int_{d_{th,1}}^{\infty}\int_{\max\{h_H,\mathcal{T}_{O}(d_{I,1})\}}^{\infty}\\ &\overline{F}_{W_{I,1}}\big(g_I^{h}(d_{O,1},d_{I,1})\big) f(d_{O,1})f(d_{I,1})\;\dd d_{O,1}\dd d_{I,1},
\end{array}
\end{equation}
with
    $g_{I}^{h}(d_{O,1},d_{I,1})= \frac{\tau}{p_{I,0}}d_{I,1}^{\alpha_{I}} \big[N_{0}+I(d_{O,1},d_{I,1}) \big]$.

The interference is approximated using:
\begin{equation}
\begin{array}{r@{}l}
    I&(d_{O,1},d_{I,1})\approx p_{O,0}(2b_{O}+\Omega_{O})\frac{(2\pi-2\theta)\lambda_{H}}{\alpha_{O}-2}z^{2-\alpha_{O}} \big|_{\infty}^{d_{O,1}}\\
    & \;\;\;+p_{O,0}(2b_{O}+\Omega_{O})\frac{2\theta\lambda_{H}}{\alpha_{O}-2}z^{2-\alpha_O} \big|_{d_{th,1}}^{\min\{d_{O,1},d_{th,1}\}}\\
    & \;\;\;+p_{I,0}(2b_{I}+\Omega_{I})\frac{2\theta\lambda_{H}}{\alpha_{I}-2}z^{2-\alpha_{I}} \big|_{\infty}^{d_{I,1}}.
\end{array}
\end{equation}

\label{theo:approx1_for_SA}
\end{theorem}
\begin{IEEEproof}
    See Appendix~\ref{app:approx1_for_SA}.
\end{IEEEproof}

Next, we introduce another approximation by overestimating the Rician region rather than the shadowed Rician region.

\begin{figure}[ht]
    \centering
    \includegraphics[width=0.85\linewidth]{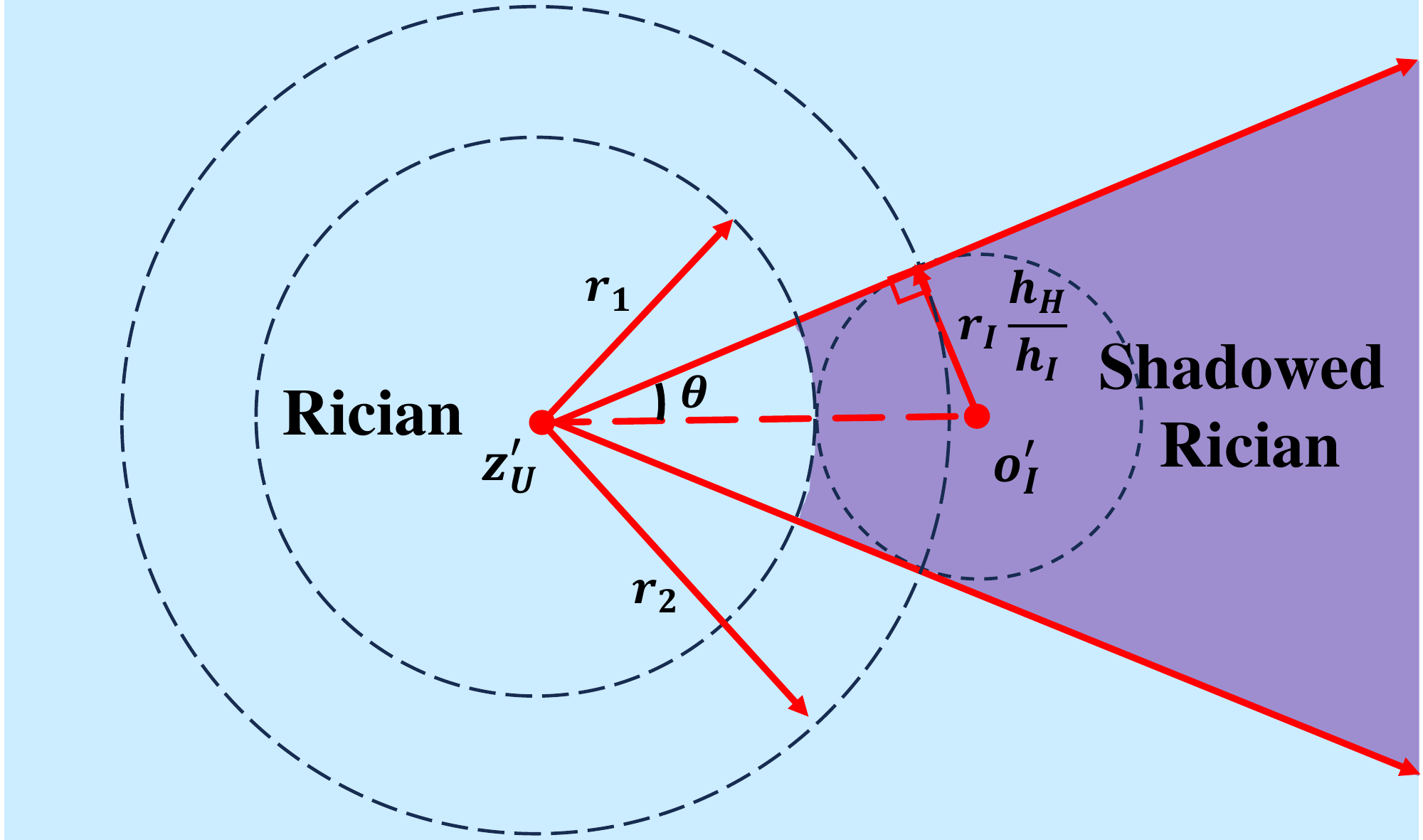}
    \caption{Approximate method 1: overestimating the shadowed Rician region.}
    \label{fig:Approx1}
\end{figure}

\subsection{Overestimating the Rician region}
We have proposed an approximation by overestimating the area of the shadowed Rician region. Next, we introduce another approximation by overestimating the area of the Rician region.\\
\textbf{Approximation Method 2:} As shown in Fig. \ref{fig:Approx2}, we overestimate the area of the Rician region. From the perspective of $\textbf{z}_U'$, the shadowed Rician region is an annular sector area centered at $\textbf{z}_U'$, with inner radius $r_2=\frac{h_H}{h_I}\sqrt{r_U^2-r_I^2}$, infinite outer radius and central angle $2\theta$. Therefore, the Rician region can be divided into: (i) a sector area centered at $\textbf{z}_U'$, with infinite radius and central angle $2\pi-2\theta$, and (ii) a sector area centered at $\textbf{z}_U'$, with radius $r_2$ and central angle $2\theta$. From the perspective of $\textbf{z}_U$, the critical value of the distance to HAPSs is $d_{th,2}=\sqrt{r_2^2+h_H^2}$.

\begin{figure}[ht]
    \centering
    \includegraphics[width=0.85\linewidth]{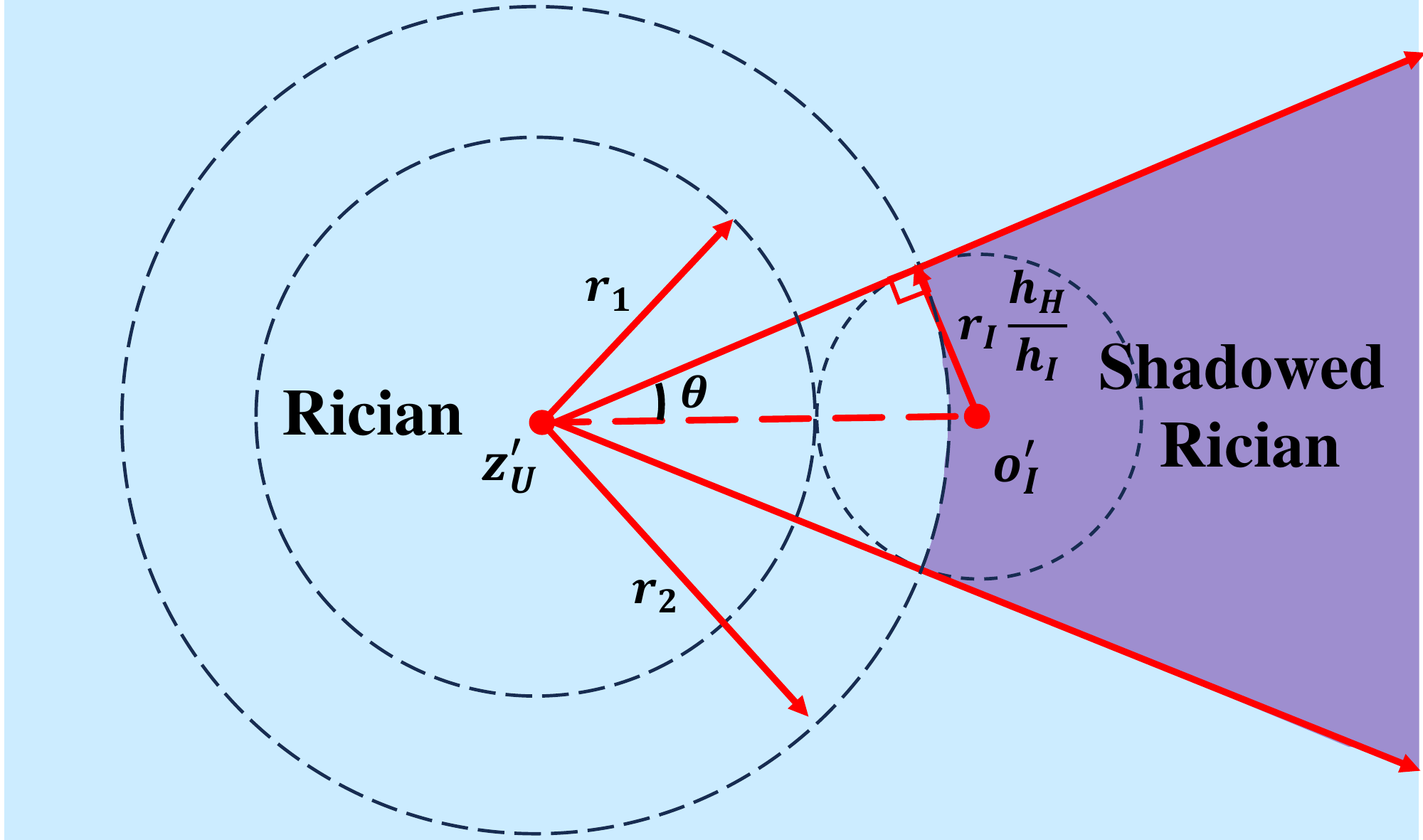}
    \caption{Approximate method 2: overestimating the Rician region.}
    \label{fig:Approx2}
\end{figure}
In this case, we also approximate the CDFs and PDFs of $d_{O,1}$ and $d_{I,1}$ as formulated in Lemma \ref{lem:CDFPDFofdO1dI1_2}.
\begin{lemma}
    Based on the approximation method 2, the CDF of the distance $d_{O,1}$ between the typical nearshore user and its closest HAPS in the Rician region is:
\begin{equation}
\begin{array}{r@{}l}
    F&(d_{O,1})\\
    &=1-\exp(-\lambda_H (\pi-\theta)d_{O,1}^2-\lambda_{H}\theta\min(d_{O,1},d_{th,2})^2)\\
    &\times \exp(\lambda_H \pi h_H^2),
\end{array}
\end{equation}
and the PDF of $d_{O,1}$ is:
\begin{equation}
\begin{array}{r@{}l}
f&(d_{O,1})=2\lambda_H d_{O,1} \big((\pi-\theta)+\theta\mathbbm{1}(d_{O,1}<d_{th,2}) \big)\\
&\times\big(1-F(d_{O,1})\big)\\
&=2\lambda_H d_{O,1} \big((\pi-\theta)+\theta \mathbbm{1}(d_{O,1}<d_{th,2}) \big)\\
&\times\exp(-\lambda_H (\pi-\theta)d_{O,1}^2-\lambda_{H}\theta\min(d_{O,1},d_{th,2})^2)\\
&\times\exp(\lambda_{H}\pi h_H^2),
\end{array}
\end{equation}
for $h_H<d_{O,1}<\infty$.

The CDF of the distance $d_{I,1}$ between the typical nearshore user and its closest HAPS in the shadowed Rician region is:
\begin{equation}
\begin{array}{r@{}l}
F(d_{I,1})=1-\exp(-\lambda_{H}\theta(d_{I,1}^2-d_{th,2}^2)),
\end{array}
\end{equation}
and the PDF of $d_{I,1}$ is
\begin{equation}
\begin{array}{r@{}l}
f(d_{I,1})&=2\lambda_{H}\theta d_{I,1}(1-F(d_{I,1}))\\
&=2\lambda_{H}\theta d_{I,1}\exp(-\lambda_{H}\theta(d_{I,1}^2-d_{th,2})),
\end{array}
\end{equation}
for $d_{th,2}<d_{I,1}<\infty$.
\label{lem:CDFPDFofdO1dI1_2}
\end{lemma}
\begin{IEEEproof}
    Similar to Lemma \ref{lem:CDFPDFofdO1dI1}, but replace $d_{th,1}$ with $d_{th,2}$.
\end{IEEEproof}
Based on the approximate CDFs and PDFs of $d_{O,1}$ and $d_{I,1}$, the approximation of coverage probability is provided in Theorem \ref{theo:approx2_for_SA}.
\begin{theorem}
    By overestimating the Rician area, the coverage probability for nearshore users can be represented using the sum of the coverage probabilities with $\textbf{H}_O$ and $\textbf{H}_I$, that is:
\begin{equation}
    P_{\rm cov}(r_U,\tau)=P_{\rm cov}(r_U,\tau,\textbf{H}_O)+P_{\rm cov}(r_U,\tau,\textbf{H}_I),
\end{equation}
where $d_{O,1}>h_H$ and $d_{I,1}>d_{th,2}$.

The coverage probability of nearshore users with $\textbf{H}_O$ is
\begin{equation}
\begin{array}{r@{}l}
    P&_{\rm cov}(r_U,\tau,\textbf{H}_O)= \int_{h_H}^{\infty} \int_{\max\{d_{th,2},\mathcal{T}_{I}(d_{O,1})\}}^{\infty}\\ &\overline{F}_{W_{O,1}}\big(g_{O}^{h}(d_{O,1},d_{I,1})\big)f(d_{O,1})f(d_{I,1})\;\dd d_{I,1}\dd d_{O,1},
\end{array}
\end{equation}
with
    $g_{O}^{h}(d_{O,1},d_{I,1})= \frac{\tau}{p_{O,0}}d_{O,1}^{\alpha_{O}} \big[N_{0}+I(d_{O,1},d_{I,1}) \big]$.

The coverage probability of nearshore users with $\textbf{H}_I$ is

\begin{equation}
\begin{array}{r@{}l}
    P&_{\rm cov}(r_U,\tau,\textbf{H}_I)= \int_{d_{th,2}}^{\infty}\int_{\max\{h_H,\mathcal{T}_{O}(d_{I,1})\}}^{\infty}\\ &\overline{F}_{W_{I,1}}\big(g_I^{h}(d_{O,1},d_{I,1})\big) f(d_{O,1})f(d_{I,1})\;\dd d_{O,1}\dd d_{I,1},
\end{array}
\end{equation}
with
    $g_{I}^{h}(d_{O,1},d_{I,1})= \frac{\tau}{p_{I,0}}d_{I,1}^{\alpha_{I}} \big[N_{0}+I(d_{O,1},d_{I,1}) \big]$.

The interference is approximated using:
\begin{equation}
\begin{array}{r@{}l}
    I&(d_{O,1},d_{I,1})\approx p_{O,0}(2b_{O}+\Omega_{O})\frac{(2\pi-2\theta)\lambda_{H}}{\alpha_{O}-2}z^{2-\alpha_{O}} \big|_{\infty}^{d_{O,1}}\\
    & \;\;\;+p_{O,0}(2b_{O}+\Omega_{O})\frac{2\theta\lambda_{H}}{\alpha_{O}-2}z^{2-\alpha_O} \big|_{d_{th,2}}^{\min\{d_{O,1},d_{th,2}\}}\\
    & \;\;\;+p_{I,0}(2b_{I}+\Omega_{I})\frac{2\theta\lambda_{H}}{\alpha_{I}-2}z^{2-\alpha_{I}} \big|_{\infty}^{d_{I,1}}.
\end{array}
\end{equation}

\label{theo:approx2_for_SA}
\end{theorem}
\begin{IEEEproof}
    Similar to Theorem \ref{theo:approx1_for_SA}, but replace $d_{th,1}$ using $d_{th,2}$.
\end{IEEEproof}

Different from the Approximation 1, Approximation 2 is closer to the offshore coverage performance because the Rician region is overestimated.

\section{Simulation results and discussion}\label{sec:simulation}
In this paper, we derive the exact expressions of downlink coverage probabilities for typical onshore, offshore and nearshore users, in Theorem \ref{theo:onshore}, \ref{theo:offshore} and \ref{theo:accurate_for_SA}, respectively, where the typical users select the HAPS with strongest average received power to access. We derive the distribution of the distance between the user and its nearest HAPS in Lemma \ref{lem:PDFofD} for the onshore and offshore users. We also derive the distributions of the distances between the user and its nearest HAPSs operating in Rician channel or shadowed-Rician channel in Lemma \ref{lem:PDFofdO1} and Lemma \ref{lem:PDFofdI1}, respectively, for the nearshore performance analysis. We also propose approximations for nearshore users in Theorem \ref{theo:approx1_for_SA} and \ref{theo:approx2_for_SA}, by overestimating the shadowed Rician region or overestimating the Rician region, respectively, where the distance distribution for onshore and offshore users are approximated in Lemma \ref{lem:CDFPDFofdO1dI1} and \ref{lem:CDFPDFofdO1dI1_2}. Inspired by the islands and maritime applications in the Caribbean, we assume that HAPSs are deployed in a $3000\;{\rm km}\times 3000\;{\rm km}$ area and conduct 20,000 Monte Carlo simulations. The altitude of the typical user is $0\;{\rm m}$ and the altitude of HAPSs is $h_H=20\;{\rm km}$. We assume that the radius of an island is $20\;{\rm km}$ and the altitude of island buildings is $h_I=20\;{\rm m}$. The transmit power is $p_t=45\;{\rm dBm}$, the carrier frequency is $f_c=2\;{\rm GHz}$ and the speed of light is $c=3.0\times 10^{8}\;{\rm m/s}$. The transmit gain is $g_t=41\;{\rm dB}$ and the receive gain is $g_r=10\;{\rm dB}$. The atmospheric attenuation parameters are $s_I=s_O=2\;{\rm dB}$ and the noise power is $N_0=1\times 10^{-20}\;{\rm W}$. We assume that $\alpha_I=4.0$, $m_I=10$, $\Omega_I=0.835$, and $b_I=0.126$ for the shadowed Rician channels, and $\alpha_O=2.3$, $\Omega_O=1.000$, and $b_O=0.200$ for the Rician channels \cite{10040542,10553646,10794418,10689625}.

\begin{figure*}[htbp]
\centering
\begin{minipage}{0.32\textwidth}
\centering 
\includegraphics[width=1\linewidth]{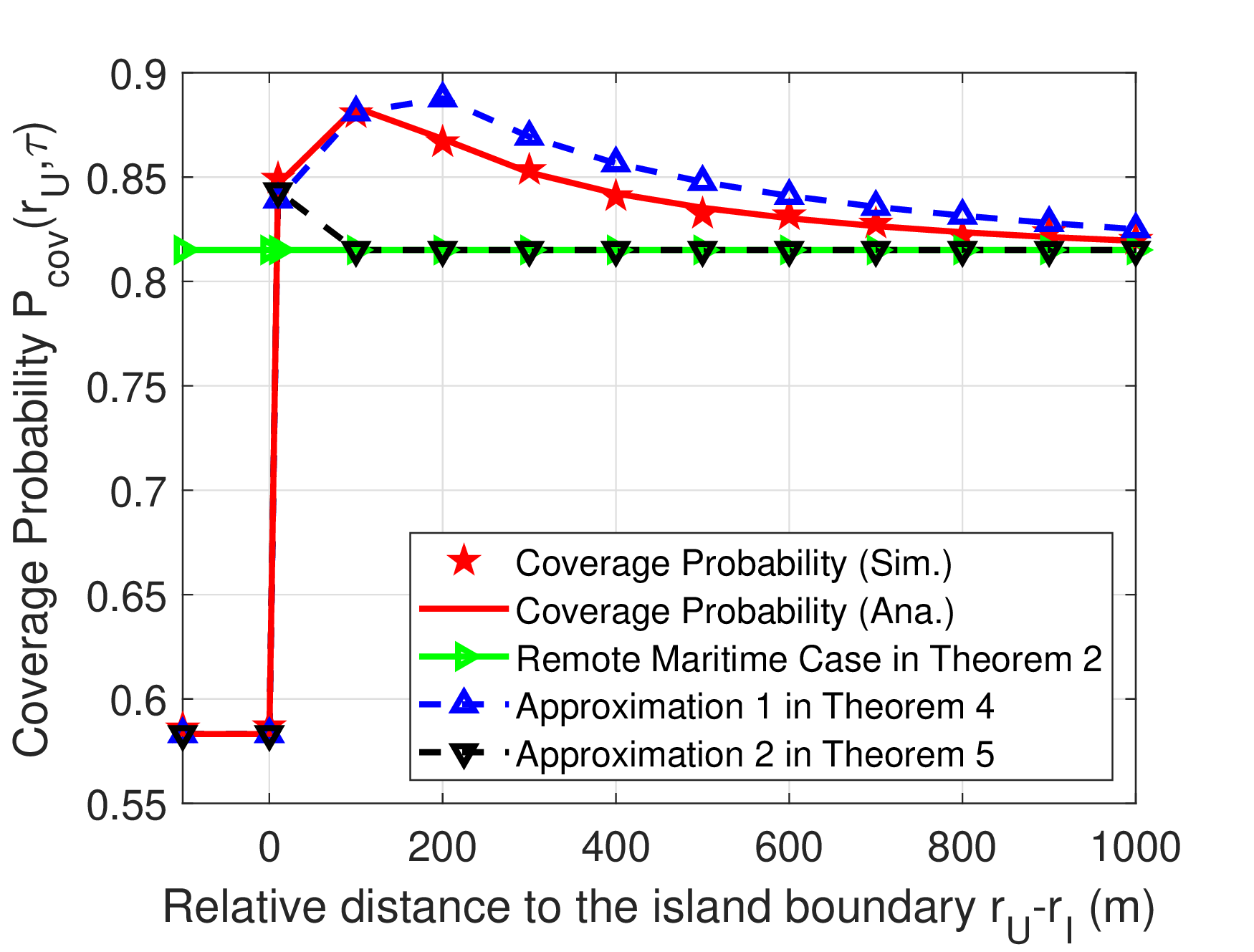}
\caption{The coverage probability considering different relative distance to the island boundary $r_U-r_I$, compared to the approximations and coverage performance of remote maritime users, when HAPS density is $3\;{\rm HAPSs}\;{\rm per}\;10^{5}\;{\rm km^2}$.}
\label{fig:Covpro_RU_lam_3e-11_Strongest}
\end{minipage}
\hfill
\begin{minipage}{0.32\textwidth}
\centering
\includegraphics[width=1\linewidth]{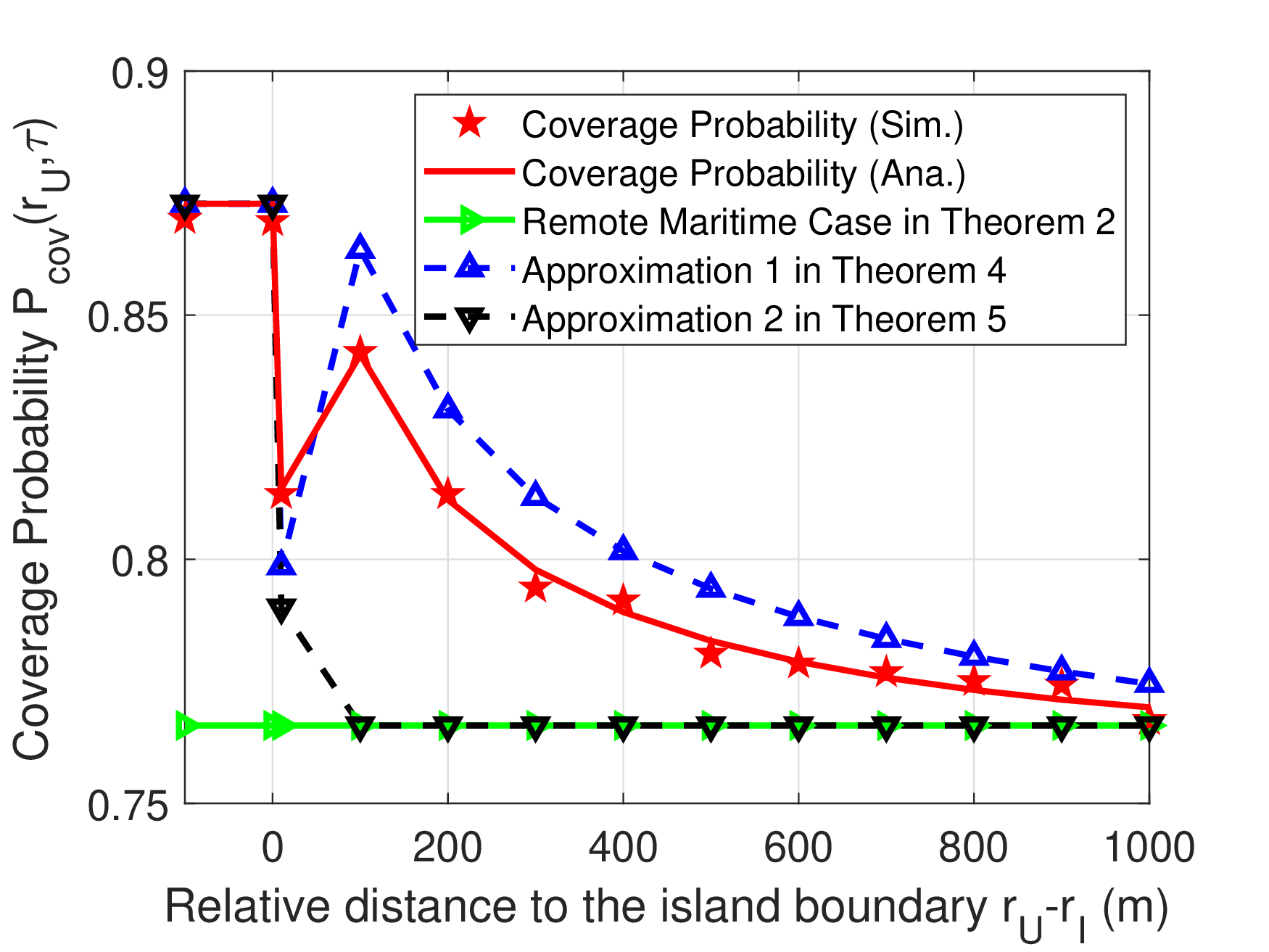}
\caption{The coverage probability considering different relative distance to the island boundary $r_U-r_I$, compared to the approximations and coverage performance of remote maritime users, when HAPS density is $1\;{\rm HAPSs}\;{\rm per}\;10^4\;{\rm km}^2$.}
\label{fig:Covpro_RU_lam_10e-11_Strongest}
\end{minipage}
\hfill
\begin{minipage}{0.32\textwidth}
\centering
\includegraphics[width=1\linewidth]{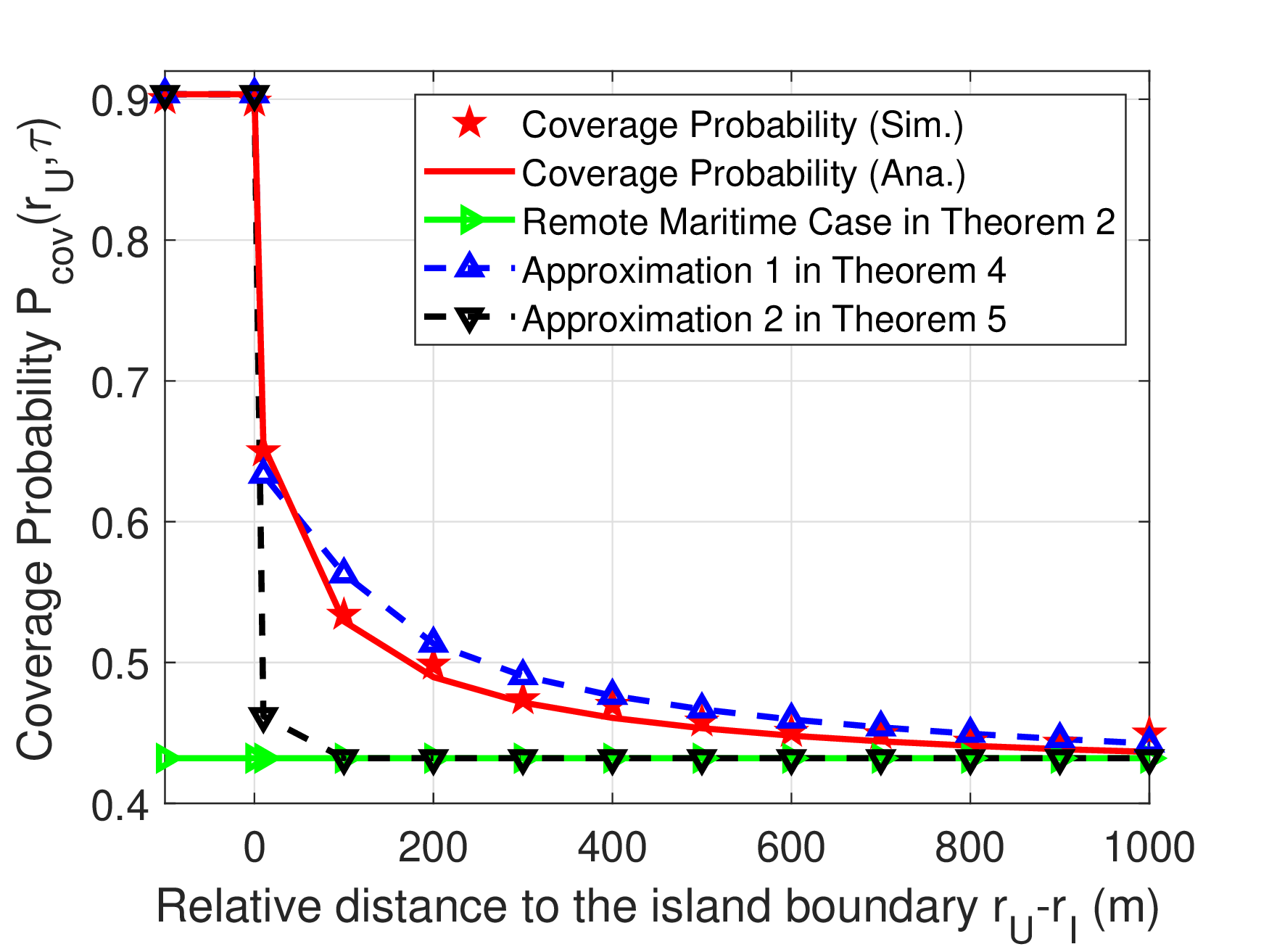}
\caption{The coverage probability considering different relative distance to the island boundary $r_U-r_I$, compared to the approximations and coverage performance of remote maritime users, when HAPS density is $1\;{\rm HAPSs}\;{\rm per}\;10^3\;{\rm km}^2$.}
\label{fig:Covpro_RU_lam_10e-10_Strongest}
\end{minipage}
\end{figure*}

In Fig. \ref{fig:Covpro_RU_lam_3e-11_Strongest}, we compare the simulated coverage probability with the approximations. We assume that $\tau=-10\;{\rm dB}$ and $\lambda_{H}=3\;{\rm HAPSs}\;{\rm per}\;10^{5}\;{\rm km^2}$. When $r_U<r_I$, the typical user is onshore and its coverage performance can be evaluated using (\ref{PcovOnshore}) in Theorem \ref{theo:onshore}, where $P_{\rm cov}(r_U,\tau)$ is 0.58. The coverage performance increases to 0.82 because more than half of HAPSs obtain the LoS condition to the user. The closest HAPS operating in the Rician channel can provide higher signal power and other HAPS operating in the Rician channel do not lead to high interference when the HAPS density is low. When $r_I<r_U<r_I+100\;{\rm m}$ and $r_U$ increases, the probability that the typical user associates to HAPSs via the Rician channel increases, so that the coverage probability increases to 0.86. When $r_U>r_I+100\;{\rm m}$ and $r_U$ continues to increase, more HAPSs have LoS conditions, which leads to increased interference and decreased coverage performance. When $r_U$ increases, the coverage probability $P_{\rm cov}(r_U,\tau)$ approaches 0.79, where almost all HAPSs have Rician channels between the typical user. The analytical curve matches the simulation results. In addition, the simulation result is closer to the Approximation 1 in Theorem \ref{theo:approx1_for_SA} than the Approximation 2 in Theorem \ref{theo:approx2_for_SA}, while both of them require lower computation time than the exact expression. Approximation 2 is close to the offshore case when $r_U>r_I$ because we overestimate the Rician region.

In Fig. \ref{fig:Covpro_RU_lam_10e-11_Strongest}, we show the coverage performance when the HAPS density is $\lambda_{H}=1\;{\rm HAPSs}\;{\rm per}\;10^{4}\;{\rm km^2}$. Unlike the low HAPS density case, onshore users have better coverage probability than offshore users. When the typical user crosses the boundary of the island, the coverage probability decreases from 0.87 to 0.81 because HAPSs operating in Rician channels leads to high interference. When $r_I<r_U<r_I+100\;{\rm m}$ and $r_U$ continues to increase, the probability that the typical user associates with the closest HAPS operating in the Rician channel increases so the coverage probability also increases. When $r_U>r_I+100\;{\rm m}$, the coverage probability decreases due to increased interference caused by the increasing density of HAPSs. As shown in Fig. \ref{fig:Covpro_RU_lam_10e-10_Strongest}, when the HAPS density is $\lambda_{H}=1\;{\rm HAPSs}\;{\rm per}\;10^{3}\;{\rm km^2}$, as the relative distance to the island boundary increases, the interference increases when more HAPSs have better LoS conditions, so that the coverage probability continues to decrease.

\begin{figure*}
\centering
    \begin{minipage}{0.32\textwidth}
    \centering
    \includegraphics[width=1\linewidth]{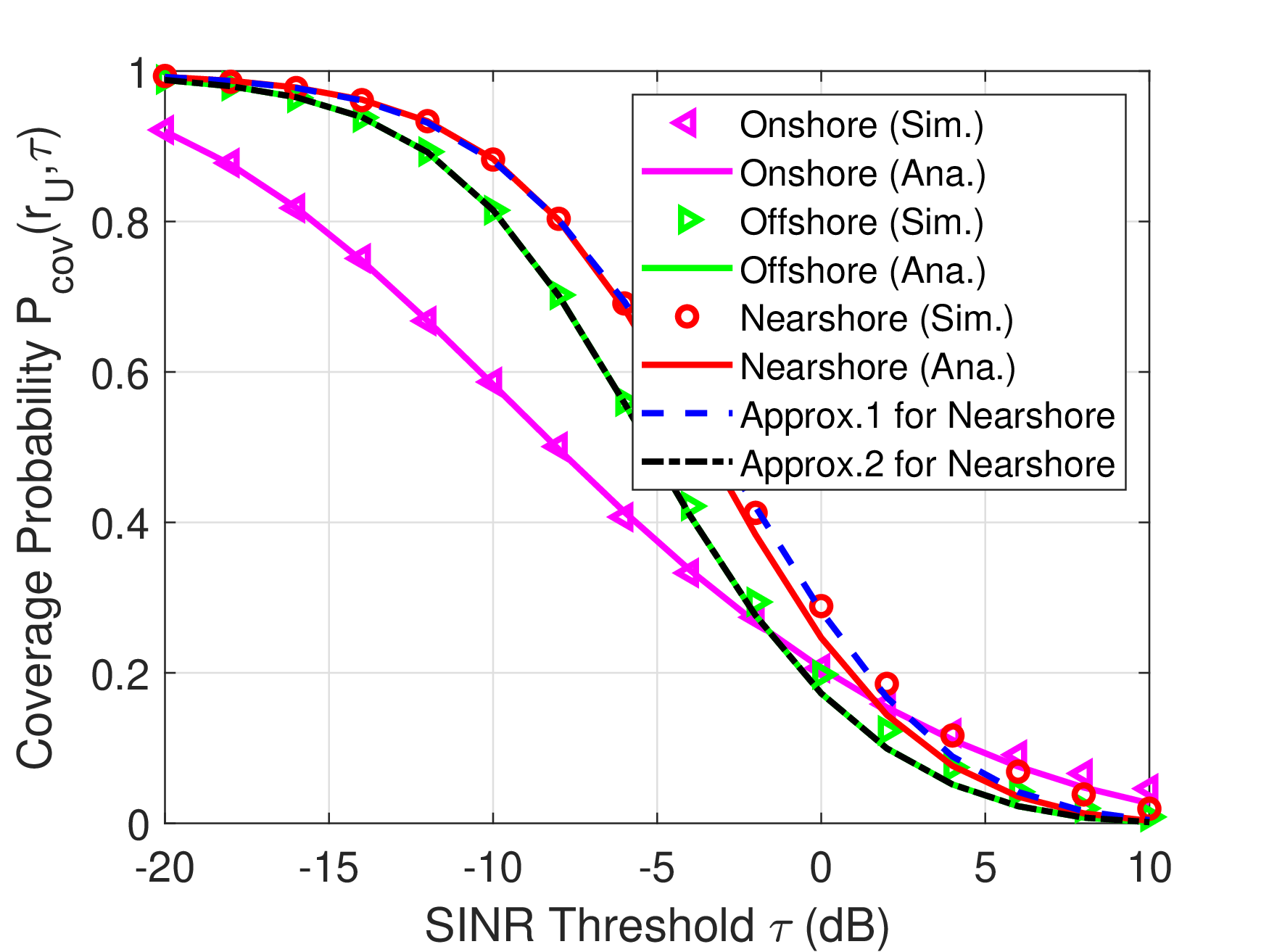}
    \caption{The CCDF of SINR threshold when $\lambda_{H}=3\;{\rm HAPSs}\;{\rm per}\;10^5\;{\rm km}^2$, considering the onshore case ($r_U=r_I-100\;{\rm m}$), nearshore case ($r_U=r_I+100\;{\rm m}$), and offshore case ($r_U=r_I+1000\;{\rm m}$).}
    \label{fig:Covpro_Beta_lam_3e-11_Strongest}
    \end{minipage}
    \hfill
    \begin{minipage}{0.32\textwidth}
    \centering
    \includegraphics[width=1\linewidth]{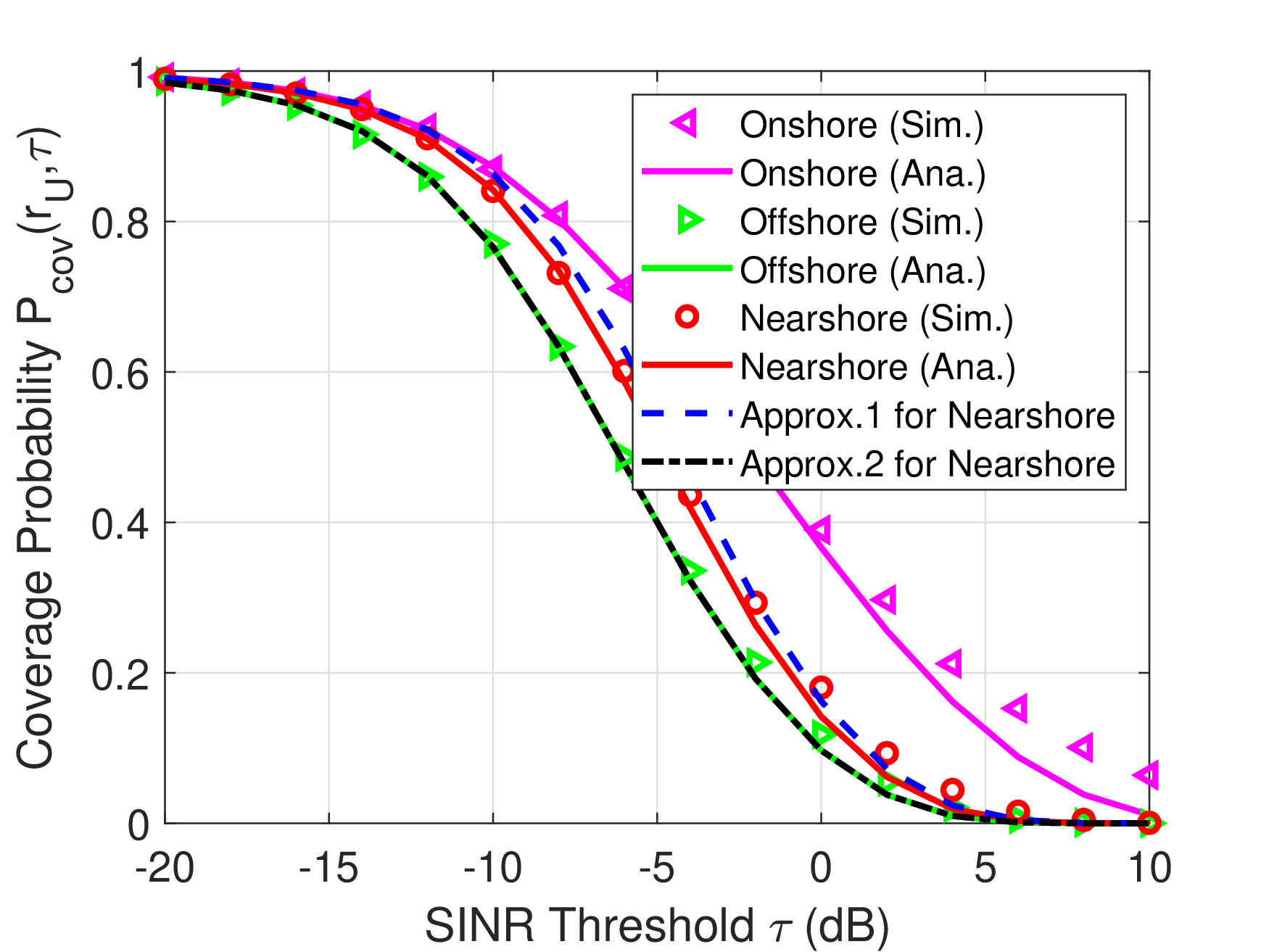}
    \caption{The CCDF of SINR threshold when $\lambda_{H}=1\;{\rm HAPSs}\;{\rm per}\;10^4\;{\rm km}^2$, considering the onshore case ($r_U=r_I-100\;{\rm m}$), nearshore case ($r_U=r_I+100\;{\rm m}$), and offshore case ($r_U=r_I+1000\;{\rm m}$).}
    \label{fig:Covpro_Beta_lam_10e-11_Strongest}
    \end{minipage}
    \hfill
    \begin{minipage}{0.32\textwidth}
    \centering
    \includegraphics[width=1\linewidth]{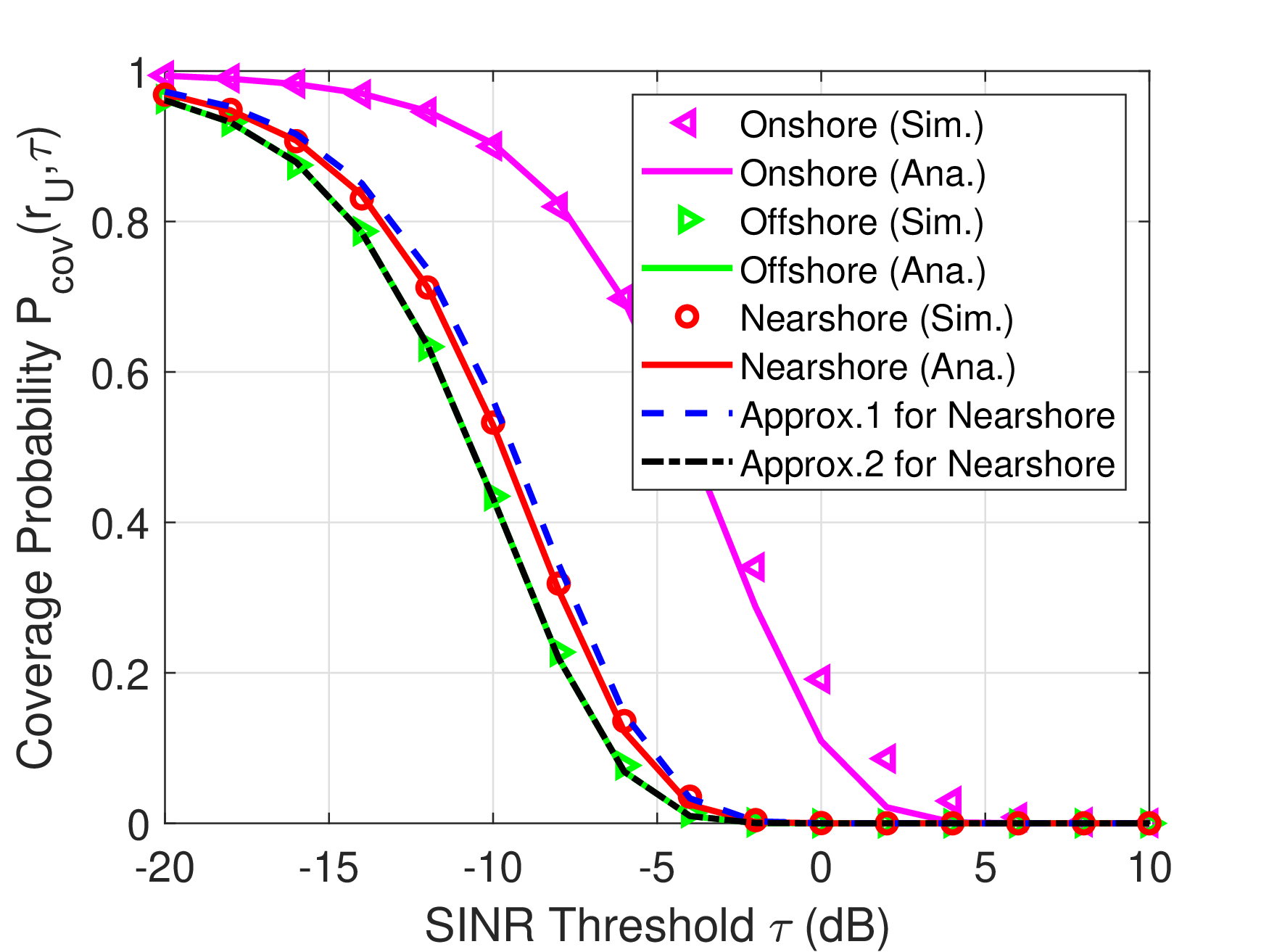}
    \caption{The CCDF of SINR threshold when $\lambda_{H}=1\;{\rm HAPSs}\;{\rm per}\;10^3\;{\rm km}^2$, considering the onshore case ($r_U=r_I-100\;{\rm m}$), nearshore case ($r_U=r_I+100\;{\rm m}$), and offshore case ($r_U=r_I+1000\;{\rm m}$).}
    \label{fig:Covpro_Beta_lam_10e-10_Strongest}
    \end{minipage}
\end{figure*}

In Fig. \ref{fig:Covpro_Beta_lam_3e-11_Strongest}, Fig. \ref{fig:Covpro_Beta_lam_10e-11_Strongest} and Fig. \ref{fig:Covpro_Beta_lam_10e-10_Strongest}, we show the relationship between coverage probability $P_{\rm cov}(r_U,\tau)$ and the SINR threshold $\tau$, considering different HAPS densities. When $\lambda_{H}=3\;{\rm HAPSs}\;{\rm per}\;10^{5}\;{\rm km^2}$ and $\tau<-2\;{\rm dB}$, offshore users and nearshore users with $r_U=r_I+100\;{\rm m}$ have better coverage performance than onshore users. However, when $\tau>2\;{\rm dB}$, onshore users have better coverage performance than offshore users and nearshore users. This is because the increase in $\tau$ makes the coverage probability more sensitive to the value of interference, but the lack of LoS conditions and the higher path loss exponent lead to lower interference than onshore and nearshore cases. When $\lambda_{H}=1\;{\rm HAPSs}\;{\rm per}\;10^{4}\;{\rm km^2}$, the onshore users have better coverage performance than offshore users and nearshore users. Moreover, when the HAPS density is $\lambda_{H}=1\;{\rm HAPSs}\;{\rm per}\;10^{3}\;{\rm km^2}$, onshore users have much better coverage performance than offshore users and nearshore users. For example, when $-10\;{\rm dB}<\tau<-5\;{\rm dB}$, onshore coverage probability is between $0.6$ and $0.9$, but the offshore and nearshore coverage probabilities are between $0.05$ and $0.5$.

\begin{figure}
    \centering
    \includegraphics[width=0.85\linewidth]{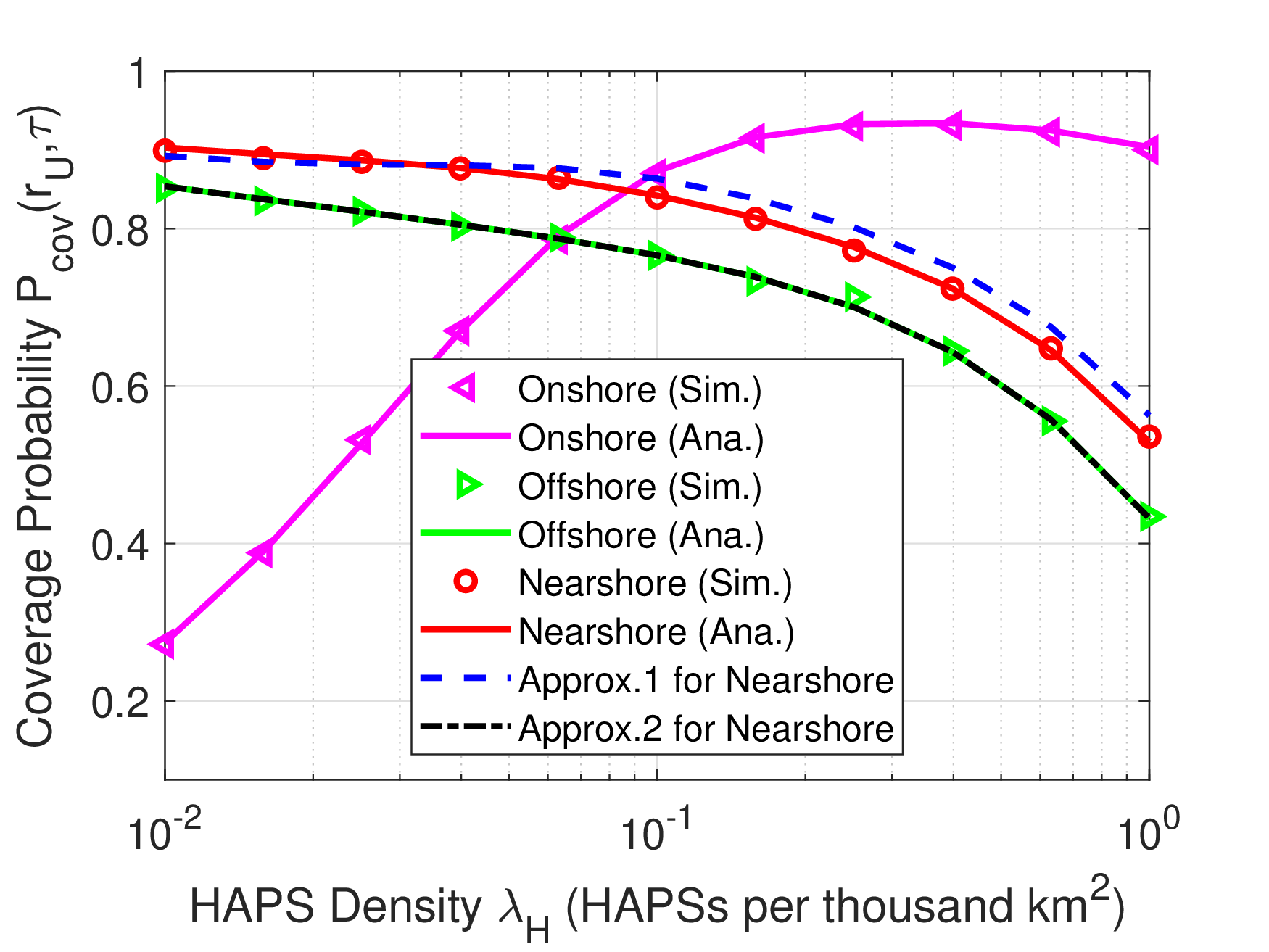}
    \caption{The relationship between coverage probability and HAPS density, considering the onshore case ($r_U=r_I-100\;{\rm m}$), nearshore case ($r_U=r_I+100\;{\rm m}$), and offshore case ($r_U=r_I+1000\;{\rm m}$).}
    \label{fig:Covpro_Dens_Strongest_log}
\end{figure}

 As shown in Fig. \ref{fig:Covpro_Dens_Strongest_log}, we show the relationship between HAPS density $\lambda_{H}$ and the coverage probability of onshore, nearshore, and offshore users. When $\lambda_{H}>8.70\;{\rm HAPSs}\;{\rm per}\;10^{5}\;{\rm km^2}$, offshore users have a better coverage performance than nearshore and offshore users. Onshore users can obtain optimal coverage performance when the HAPS density is around $\lambda_{H}=4\;{\rm HAPSs}\;{\rm per}\;10^{4}\;{\rm km^2}$, because a higher density can help the typical user connect to closer HAPSs. Unlike onshore users, offshore users have better coverage performance than onshore users when $\lambda_{H}<6.31\;{\rm HAPSs}\;{\rm per}\;10^{5}\;{\rm km^2}$. The coverage probability can achieve 0.9 when the HAPS density is close to $1\;{\rm HAPSs}\;{\rm per}\;10^5\;{\rm km}^2$, and the increase in the HAPS density leads to higher interference to the typical user. The increase in HAPS density also leads to a decrease in the coverage probability of nearshore users with $r_U=r_I+100\;{\rm m}$. 
 
As shown in the results in Fig. 6-12, Approximation 2 is close to the offshore performance because the Rician region is overestimated and the hybrid channel environment in a finite simulation area is almost reduced to a Rician environment like the offshore case. Approximation 1 is closer to the actual coverage curve because the channel conditions in fewer areas are misestimated. Two approximations can help evaluate the range of nearshore coverage performance with less waste of computational resource sensing information storage of performance analysis when the accuracy is not strict. They only require the angle range, altitude and minimum/maximum distance to the blockage clusters. Our proposed mathematical framework for coverage analysis can help provide reference to estimate coverage performance and build digital twins. As shown in Fig. \ref{fig:jamaica1}, we took Jamaica as an example. Without HAPS deployment, the existing terrestrial network of Flow already cover most of areas. However, the deep forest areas and deep sea areas still lack of coverage. When the HAPS density in the whole Caribbean Zone is $1\;{\rm HAPSs\;per\;10^4\;km^2}$, there is one or two HAPSs will be deployed upon the Jamaica. Our proposed methods in Theorem \ref{theo:onshore} and Theorem \ref{theo:offshore} can help evaluate the coverage performance, where the coverage probabilities of offshore and onshore users are higher than 0.75. Complex terrain make the nearshore coverage performance analysis more difficult than the ideal model. As shown in Fig. \ref{fig:jamaica2}, through radar or camera, the considered user can sense the obstruction of forest or mountains within a range of about 90 degrees. The main terrain affecting the LoS paths between the considered user and HAPSs is on the circular area with radius $r_I=35\;{\rm km}$ and altitude $h_I=700\;{\rm m}$. The distance between the user and this circular area is $r_U=r_I+15\;{\rm km}$. Adopting these parameters, the coverage probability of the considered nearshore user can be approximated using Theorem \ref{theo:approx1_for_SA} and Theorem \ref{theo:approx2_for_SA}.

\begin{figure}
    \centering
    \includegraphics[width=0.85\linewidth]{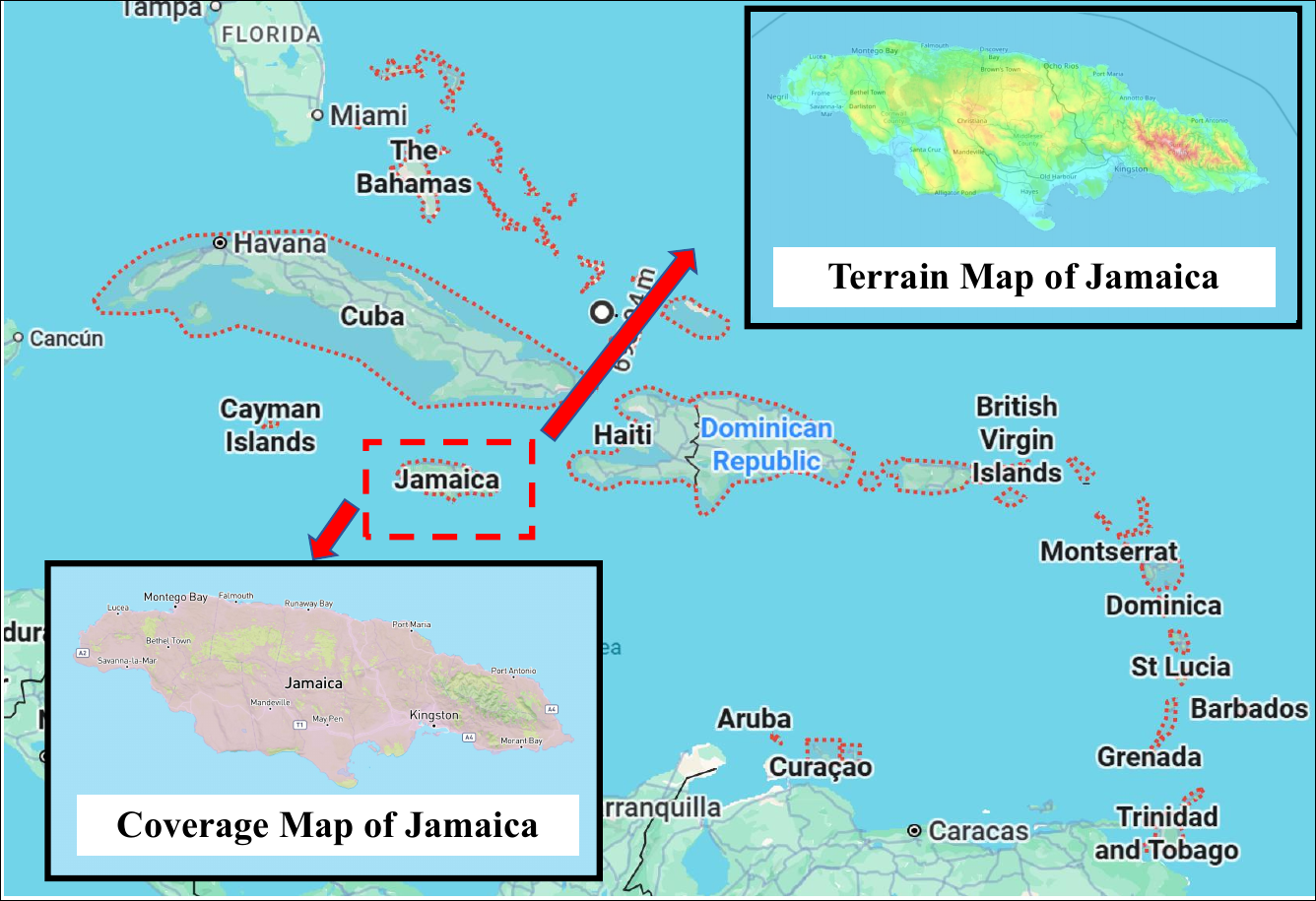}
    \caption{Coverage map (from GSMA) and terrain map of Jamaica.}
    \label{fig:jamaica1}
\end{figure}

\begin{figure}
    \centering
    \includegraphics[width=0.85\linewidth]{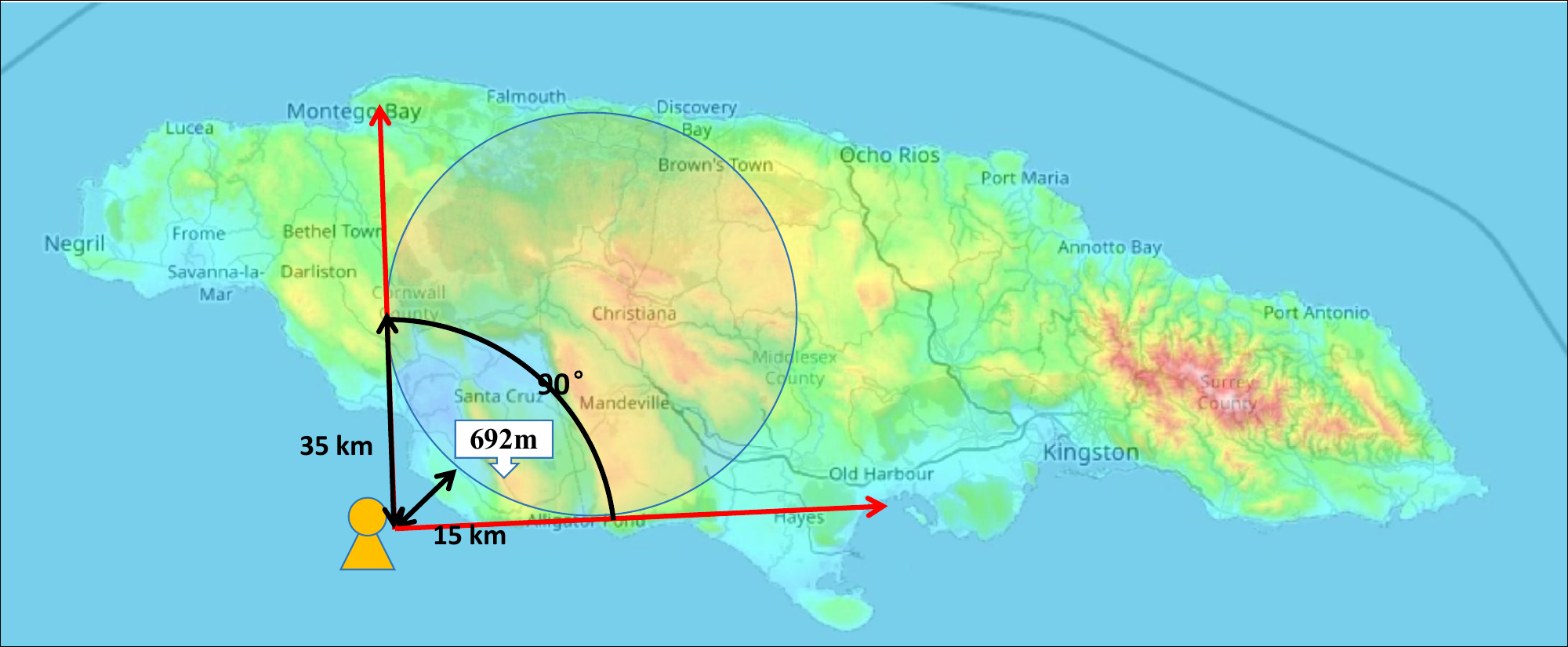}
    \caption{Approximation methods for nearshore performance analysis on Jamaica.}
    \label{fig:jamaica2}
\end{figure}

However, in this paper, some ideal assumptions of shadowing effect, reflections and ducting effect are adopted, to first build a mathematical framework evaluating the effect of user location, HAPSs' location and finite shadowing zones on the downlink performance. In the future, more complex electromagnetic environment analysis can be realized and visualized by digital twins with ray tracing technique and sensing networks. It is worth noting that, in this paper, we assume that the HAPSs interfere with each other without beamforming optimization, so that the coverage performance does not exceed 0.9. With the massive antenna array and the decoding technique, the overall performance of the HAPS network can be improved. The results shown in this paper can work as baseline or lower bound and the evaluation approaches can still work by further taking the beamforming gains into consideration. Our analysis and numerical results mainly show how to ensure user fairness by customizing the HAPS deployment. By comparing the onshore, nearshore and offshore performance, we propose to deploy less HAPSs in open areas without severe shadowing effect, to let each HAPS cover as large maritime areas as possible, but carefully select a higher HAPS density to cover onshore and nearshore users facing the shadowing effect. Based on this, the HAPS operators can design the customized HAPS deployment more easily, and realize user fairness. Furthermore, the uplink analysis and joint analysis are expected to be conducted in the future, to satisfy different applications with various user distribution and communication requirements. 

\section{Conclusion}\label{sec:conclusion}
In this paper, we introduced the evaluation methods of HAPS-based solutions for mobile users in islands or maritime regions. We evaluated the expressions of the coverage probabilities for onshore, nearshore, and offshore users. For nearshore users, we proposed two approximations, by overestimating the shadowed Rician region (Approximation 1) or overestimating the Rician region (Approximation 2), from which we can reduce the computation complexity. We verified that onshore users require a higher HAPS density, and offshore users require a lower HAPS density. Approximation 1 works better than Approximation 2, and both of them can help to evaluate the effect of any complex shape of islands on nearshore users. We also showed how the decoding SINR threshold and HAPS density affect the coverage performance of onshore, nearshore, and offshore users. The simulation results show that an inhomogeneous HAPS deployment is necessary in the future to realize higher and fair coverage performance for typical users at different geographical locations.

\appendices
\section{Proof of Lemma \ref{lem:PDFofD}}\label{app:PDFofD}
In this paper, we assume that HAPSs are uniformly deployed at altitude $h_H$, where the set of their locations follows a 2D PPP with density $\lambda_{H}$. We assume that $\textbf{z}_{U}=(-r_{U},0,0)$ is the location of the typical user and $\textbf{z}_{U}'=(-r_{U},0,h_H)$ is directly above $\textbf{z}_U$ at a height of $h_H$. The distance between $\textbf{z}_{U}'$ and its closest HAPS to the typical user is $\zeta$, and the distance between the typical user at
$\textbf{z}_{U}$ and the closest HAPS is $D$. The relationship between $D$ and $\zeta$ follows the Pythagorean theorem:
    $D^2=\zeta^2+h_H^2$.
The CDF of $\zeta$ can be expressed as:
\begin{equation}
    F_{\zeta}(\zeta)=1-\exp(-\lambda_{H}\pi\zeta^2).
\end{equation}
By taking the derivative of $F_{\zeta}(\zeta)$, the PDF of $\zeta$ is:
\begin{equation}
    f_{\zeta}(\zeta)=2\lambda_{H}\pi \zeta\exp(-\lambda_{H}\pi\zeta^2),
\end{equation}
which satisfies:
\begin{equation}
\begin{array}{r@{}l}
 \int_{0}^{\infty}&f_{\zeta}(\zeta)\dd\zeta =\int_{0}^{\infty}2\lambda_{H}\pi \zeta\exp(-\lambda_{H}\pi\zeta^2)\dd\zeta\\
& =\int_{h_H}^{\infty}2\lambda_{H}\pi \exp(\lambda_{H}\pi h_H^2)d\exp(-\lambda_{H}\pi d^2)\dd d\\
& =\int_{h_H}^{\infty}f_{D}(d)\dd d=1.
\end{array}
\end{equation}
Therefore, the PDF of $d$ is
\begin{equation}
    f_{D}(d)=2\lambda_{H}\pi\exp(\lambda_{H}\pi h_{H}^2) d\exp(-\lambda_{H}\pi d^2),
\end{equation}
for $h_H<d<\infty$.

\section{Proof of Theorem \ref{theo:onshore}}\label{app:onshore}
Notice that $r_U$ is the distance from the typical user to the center of the island, and $d$ is the distance between the user and its closest HAPS located at $\textbf{x}_1^{I}$. We use $W_{I,1}$ to represent the channel gain between them, which follows the shadowed Rician distribution shown in (\ref{fWIix}) and (\ref{FWIix}). Therefore, the coverage probability can be calculated using the conditional coverage probability and the PDF of $D$:
\begin{equation}
\begin{array}{r@{}l}
    P&_{\rm cov}(r_U,\tau)\overset{\triangle}{=}\P\{{\rm SINR}_{I}(r_U)>\tau\}\\
    & =\int_{h_H}^{\infty}\P\{{\rm SINR}_{I}(r_U)>\tau|d\}f_{D}(d){\rm d}d\\
    & =\int_{h_H}^{\infty} \P \big\{\frac{p_{I,0}W_{I,1}d^{-\alpha_I}}{N_{0}+I_{I}(d)}>\tau \big\}f_{D}(d){\rm d}d\\
    & =\int_{h_H}^{\infty} \overline{F}_{W_{I,1}}\big(g_{I}(d)\big)f_{D}(d){\rm d}d,
\end{array}
\end{equation}
where
\begin{equation}
\begin{array}{r@{}l}
    \overline{F}&_{W_{I,1}}(x)\overset{\triangle}{=} \big(\frac{2b_{I}m_{I}}{2b_{I}m_{I}+\Omega_{I}}\big)^{m_{I}}\\
    & \times\sum_{n=0}^{\infty}\frac{(m_{I})_n}{(1)_{n}n!} \big(\frac{\Omega_{I}}{2b_{I}m_{I}+\Omega_{I}} \big)^{n}\Gamma\big(n+1,\frac{x}{2b_{I}}\big).
\end{array}
\end{equation}
$\Gamma(s,x)$ is the upper incomplete Gamma function with shape parameter $s$ and $g_{I}(d)= \frac{\tau}{p_{I,0}}N_{0}d^{\alpha_{I}}+\frac{\tau}{p_{I,0}}I_{I}(d)d^{\alpha_{I}}$. The interference can be approximated using its expectation:
\begin{equation}
\begin{array}{r@{}l}
    I&_{I}(d)\approx \mathbb{E}(\sum\limits_{d_{I,k}>d} p_{I,0}W_{I,k}d_{I,k}^{-\alpha_I})\\
    & \overset{(a)}{=}p_{I,0}\mathbb{E}[W_{I,k}]\int_0^{2\pi} \int_{\sqrt{d^2-h_H^2}}^{\infty} \lambda_{H} (\zeta^2+h_H^2)^{-\frac{\alpha_I}{2}}\zeta{\rm d}\zeta {\rm d}\gamma\\
    & =p_{I,0}\mathbb{E}[W_{I,k}]\frac{2\pi\lambda_{H}}{\alpha_{I}-2}z^{2-\alpha_{I}} \big|_{\infty}^{d},
\end{array}
\end{equation}
where step (a) comes from Campbell theorem \cite{haenggi2012stochastic}.
Because the expectation of $W_{I,k}$ is:
    $\mathbb{E}[W_{I,k}]=2b_{I}+\Omega_{I}$,
so that 
    $g_{I}(d)\approx \frac{\tau N_{0}}{p_{I,0}}d^{\alpha_{I}}+\frac{2\pi\lambda_{H}\tau (2b_{I}+\Omega_{I})d^{\alpha_I}}{\alpha_{I}-2}z^{2-\alpha_I} \big|_{\infty}^{d}$.

\section{Proof of Theorem \ref{theo:offshore}}\label{app:offshore}
Notice that $r_U$ is the distance from the typical user to the center of the island, and $d$ is the distance between the user and its closest HAPS located at $x_{1}^{O}$. We use $W_{O,1}$ to represent the channel gain between them, which follows the Rician distribution in (\ref{fWOix}) and (\ref{FWOix}). Therefore, the coverage probability can be calculated using the conditional coverage probability and the PDF of $D$:
\begin{equation}
\begin{array}{r@{}l}
    P&_{\rm cov}(r_U,\tau)\overset{\triangle}{=}\P\{{\rm SINR}_{O}(r_U)>\tau\}\\
    & =\int_{h_H}^{\infty}\P\{{\rm SINR}_{O}(r_U)>\tau|d\}f_{D}(d){\rm d}d\\
    & =\int_{h_H}^{\infty} \P \big\{\frac{p_{O,0}W_{O,1}d^{-\alpha_O}}{N_{0}+I_{O}(d)}>\tau \big\}f_{D}(d){\rm d}d\\
    & =\int_{h_H}^{\infty} \overline{F}_{W_{O,1}}\big(g_{O}(d)\big)f_{D}(d){\rm d}d,
\end{array}
\end{equation}
where
$\overline{F} _{W_{O,1}}\big(g_{O}(d)\big)\overset{\triangle}{=}Q_{1} \big(\sqrt{\frac{\Omega_{O}}{b_O}},\sqrt{\frac{g_{O}(d)}{b_O}} \big)$.
$Q_1(a,b)$ is the Marcum Q-function of first kind of order $\nu=1$ and $g_{O}(d)= \frac{\tau}{p_{O,0}}N_{0}d^{\alpha_{O}}+\frac{\tau}{p_{O,0}}I_{O}(d)d^{\alpha_{O}}$. The interference can be approximated using its expectation:
\begin{equation}
\begin{array}{r@{}l}
    I&_{O}(d)\approx \mathbb{E}(\sum\limits_{d_{O,j}>d} p_{O,0}W_{O,j}d_{O,j}^{-\alpha_O})\\
    & \overset{(b)}{=}p_{O,0}\mathbb{E}[W_{O,j}]\int_0^{2\pi} \int_{\sqrt{d^2-h^2}}^{\infty} \lambda_{H} (\zeta^2+h_H^2)^{-\frac{\alpha_O}{2}}\zeta{\rm d}\zeta {\rm d}\gamma\\
    & =p_{O,0}\mathbb{E}[W_{O,j}]\frac{2\pi\lambda_{H}}{\alpha_{O}-2}z^{2-\alpha_{O}} \big|_{\infty}^{d},
\end{array}
\end{equation}
where step (b) comes from Campbell theorem. Because the expectation of $W_{O,j}$ is:
    $\mathbb{E}[W_{O,j}]=2b_{O}+\Omega_{O}$,
so that
    $g_{O}(d)\approx \frac{\tau N_{0}}{p_{O,0}}d^{\alpha_{O}}+\frac{2\pi\lambda_{H}\tau (2b_O+\Omega_O)d^{\alpha_O}}{\alpha_{O}-2}z^{2-\alpha_O} \big|_{\infty}^{d}$.

\section{Proof of Lemma \ref{lem:PDFofdO1}}\label{app:PDFofdO1}
To evaluate the coverage probability for nearshore users with a hybrid channel environment, we define that HAPSs in the Rician region follow a PPP with density $\lambda_H$. We define $d_{th,1}=\sqrt{r_1^2+h_H^2}$ where $r_1=\frac{h_H}{h_I}(r_U-r_I)$ and $d_{th,2}=\sqrt{r_2^2+h_H^2}$ where $r_2=\frac{h_H}{h_I}\sqrt{r_U^2-r_I^2}$. We discuss CDFs and PDFs of $d_{O,1}$ for the cases where $h_H<d_{O,1}<d_{th,1}$, $d_{th,1}<d_{O,1}<d_{th,2}$ and $d_{th,2}<d_{O,1}<\infty$, respectively. We assume that the area of the Rician region within a distance $d_{O,1}$ is $S_{O}(d_{O,1})$.

For the case where $h_H<d_{O,1}<d_{th,1}$, we have
\begin{equation}
    S_{O}(d_{O,1})=\pi d_{O,1}^2-\pi h_H^2.
\end{equation}
Therefore, the CDF of $d_{O,1}$ can be calculated using:
\begin{equation}
\begin{array}{r@{}l}
    F(d_{O,1})&=1-\exp(-\lambda_H S_O(d_{O,1}))\\
    &=1-\exp(-\lambda_{H} \pi d_{O,1}^2)\exp(\lambda_{H}\pi h_H^2),
\end{array}
\end{equation}
and the PDF of $d_{O,1}$ is the derivative of $F(d_{O,1})$, which is equal to the multiplication of HAPS density $\lambda_H$, angle range of Rician region at $d_{O,1}$ which is $2\pi$, the distance $d_{O,1}$, and the CCDF of $d_{O,1}$:
\begin{equation}
    f(d_{O,1})=2\lambda_{H}\pi d_{O,1}\exp(-\lambda_{H} \pi d_{O,1}^2)\exp(\lambda_{H}\pi h_H^2).
\end{equation}

For the case where $d_{th,1}<d_{O,1}<d_{th,2}$, we assume that the intersections of the circle centered at $\textbf{z}_U'$ with radius $\zeta=\sqrt{d_{O,1}^2-h_H^2}$, and the circle centered at $\textbf{o}_I'$ with radius $r_I\frac{h_H}{h_I}$ as $P_1$ and $P_2$. We can divide $S_{O}(d_{O,1})$ into two parts $S_1(d_{O,1})$ and $S_2(d_{O,1})$. $S_1(d_{O,1})$ is a sector area centered at $\textbf{z}_U'$ with radius $d_{O,1}$, central angle $2\pi-2\phi(d_{O,1})$, where
\begin{equation}
\begin{array}{r@{}l}
    \phi&(d_{O,1})= \arccos \big(\frac{\zeta^2+r_U^2\frac{h_H^2}{h_I^2}-r_I^2\frac{h_H^2}{h_I^2}}{2\zeta r_U\frac{h_H}{h_I}} \big)\\
    &= \arccos \big(\frac{h_I^2 d_{O,1}^2-h_I^2 h_H^2+r_U^2 h_H^2-r_I^2 h_H^2}{2\sqrt{d_{O,1}^2-h_H^2} r_U h_H h_I} \big).
\end{array}
\end{equation}
Especially, $\phi(d_{th,1})=0$ and $\phi(d_{th,2})=\arcsin{\frac{r_I}{r_U}}$. Therefore, we have
\begin{equation}
    S_1(d_{O,1})=(\pi-\phi(d_{O,1}))(d_{O,1}^2-h_H^2).
\label{S1}
\end{equation}

We obtain $S_2(d_{O,1})$ by subtracting the sector area centered at $\textbf{o}_I'$ with radius $r_I\frac{h_H}{h_I}$ and central angle $\eta(d_{O,1})$ from the quadrilateral $\textbf{z}_U'P_1\textbf{o}_I'P_2$, where
\begin{equation}
\begin{array}{r@{}l}
    \eta&(d_{O,1})= \arccos \big(\frac{r_U^2\frac{h_H^2}{h_I^2}+r_I^2\frac{h_H^2}{h_I^2}-\zeta^2}{2r_I\frac{h_H}{h_I} r_U\frac{h_H}{h_I}} \big)\\
    &= \arccos \big(\frac{r_U^2h_H^2+r_I^2h_H^2-d_{O,1}^2h_I^2+h_H^2h_I^2}{2r_I r_Uh_H^2} \big).
\end{array}
\end{equation}
Especially, $\eta(d_{th,1})=0$ and $\eta(d_{th,2})=\arccos\frac{r_I}{r_U}$. Therefore, we have:
\begin{equation}
S_2(d_{O,1})=\sin\eta(d_{O,1})r_Ir_U\frac{h_H^2}{h_I^2}-\eta(d_{O,1})r_I^2\frac{h_H^2}{h_I^2},
\label{S2}
\end{equation}
and
\begin{equation}
\begin{array}{r@{}l}
S_O(d_{O,1})&=S_1(d_{O,1})+S_2(d_{O,1})\\
&=(\pi-\phi(d_{O,1}))(d_{O,1}^2-h_H^2)\\
&+\sin\eta(d_{O,1})r_Ir_U\frac{h_H^2}{h_I^2}-\eta(d_{O,1})r_I^2\frac{h_H^2}{h_I^2}.
\end{array}
\end{equation}

The CDF can be formulated using
\begin{equation}
\begin{array}{r@{}l}
    F&(d_{O,1})=1-\exp(-\lambda_H S_O(d_{O,1}))\\
    &=1-\exp(-\lambda_{H}(\pi-\phi(d_{O,1}))(d_{O,1}^2-h_H^2))\\
    &\times\exp(-\lambda_{H}\sin\eta(d_{O,1})r_Ir_U\frac{h_H^2}{h_I^2}+\lambda_{H}\eta(d_{O,1})r_I^2\frac{h_H^2}{h_I^2}),
\end{array}
\end{equation}
and the PDF is the derivative of $F(d_{O,1})$, which is also equal to the multiplication of HAPS density $\lambda_H$, angle range of Rician region at $d_{O,1}$ which is $2\phi(d_{O,1})$, the distance $d_{O,1}$, and the CCDF of $d_{O,1}$, which is:
\begin{equation}
\begin{array}{r@{}l}
    f&(d_{O,1})=2\lambda_{H}(\pi-\phi(d_{O,1}))d_{O,1}\\
    &\times \exp(-\lambda_{H}(\pi-\phi(d_{O,1}))(d_{O,1}^2-h_H^2))\\
    &\times\exp(-\lambda_{H}\sin\eta(d_{O,1})r_Ir_U\frac{h_H^2}{h_I^2})\exp(\lambda_{H}\eta(d_{O,1})r_I^2\frac{h_H^2}{h_I^2}).
\end{array}
\end{equation}

For the case where $d_{O,1}>d_{th,2}$, $S_1(d_{O,1})$ has a central angle $\theta$ instead of $\phi(d_{O,1})$ in (\ref{S1}):
    $S_1(d_{O,1})=(\pi-\theta)(d_{O,1}^2-h_H^2)$
and $S_2(d_{O,1})$ can be obtained by subtracting the sector area centered at $\textbf{o}_I'$ with radius $r_I\frac{h_H}{h_I}$ and central angle $\frac{\pi}{2}-\theta$ from the quadrilateral $\textbf{z}_U' P_1 \textbf{o}_I' P_2$ where $\angle \textbf{z}_U' P_1 \textbf{o}_I'=\angle \textbf{z}_U' P_2 \textbf{o}_I'=\frac{\pi}{2}$. The area $S_2(d_{O,1})$ can be calculated by substituting $d_{th,2}$ into $d_{O,1}$ in (\ref{S2}), so that
    $S_2(d_{O,1})=r_I^2\frac{h_H^2}{h_I^2} \big( \cot\theta-(\frac{\pi}{2}-\theta) \big)$
and
\begin{equation}
\begin{array}{r@{}l}
    S&_O(d_{O,1})=S_1(d_{O,1})+S_2(d_{O,1})\\
    &=(\pi-\theta)(d_{O,1}^2-h_H^2)+r_I^2\frac{h_H^2}{h_I^2} \big( \cot\theta-(\frac{\pi}{2}-\theta) \big).
\end{array}
\end{equation}

The CDF can be formulated using:
\begin{equation}
\begin{array}{r@{}l}
    F&(d_{O,1})=1-\exp(-\lambda_H S_O(d_{O,1}))\\
    &=1-\exp(-\lambda_{H}(\pi-\theta)(d_{O,1}^2-h_H^2))\\
    &\times\exp(-\lambda_{H}\cot\theta r_I^2\frac{h_H^2}{h_I^2}+\lambda_{H}(\frac{\pi}{2}-\theta)r_I^2\frac{h_H^2}{h_I^2}),
\end{array}
\end{equation}
and the PDF can be formulated using:
\begin{equation}
\begin{array}{r@{}l}
    f&(d_{O,1})=2\lambda_H(\pi-\theta)d_{O,1}\\
    &\times\exp(-\lambda_{H}(\pi-\theta)(d_{O,1}^2-h_H^2))\\
    &\times\exp(-\lambda_{H}\cot\theta r_I^2\frac{h_H^2}{h_I^2}+\lambda_{H}(\frac{\pi}{2}-\theta)r_I^2\frac{h_H^2}{h_I^2}).
\end{array}
\end{equation}
We can summarize the PDF of $d_{O,1}$ as:
\begin{equation}
\begin{array}{l@{}l}
&f(d_{O,1})=2\lambda_{H}d_{O,1}(1-F(d_{O,1}))\\
&\times (\pi-\phi(d_{O,1})\mathbbm{1}(d_{th,1}<d_{O,1}<d_{th,2})-\theta\mathbbm{1}(d_{O,1}>d_{th,2})).
\end{array}
\end{equation}

\section{Proof of Lemma \ref{lem:PDFofdI1}}\label{app:PDFofdI1}
To evaluate the coverage probability for nearshore users within a hybrid channel environment, we also assume that HAPSs in the shadowed Rician region follow a PPP with density $\lambda_{H}$. We discuss CDFs and PDFs of $d_{I,1}$ for the cases where $d_{th,1}<d_{I,1}<d_{th,2}$ and $d_{th,2}<d_{I,1}<\infty$, respectively. We define that the area of the shadowed Rician region within a distance $d_{I,1}$ is $S_I(d_{I,1})$.

For the case where $d_{th,1}<d_{I,1}<d_{th,2}$, we can calculate $S_I(d_{I,1})$ using the subtracting $S_2(d_{I,1})$ from the sector area centered at $\textbf{z}_U'$ with radius $d_{O,1}$ and central angle $2\phi(d_{I,1})$, where
$S_2(d_{I,1})=\sin\eta(d_{I,1})r_Ir_U\frac{h_H^2}{h_I^2}-\eta(d_{I,1})r_I^2\frac{h_H^2}{h_I^2}$,
and
    $S_3(d_{I,1})=\phi(d_{I,1})(d_{I,1}^2-h_H^2)$,
with
    $\phi(d_{I,1})= \arccos \big(\frac{h_I^2 d_{I,1}^2-h_I^2 h_H^2+r_U^2 h_H^2-r_I^2 h_H^2}{2\sqrt{d_{I,1}^2-h_H^2} r_U h_H h_I} \big)$
and
    $\eta(d_{I,1})= \arccos \big(\frac{r_U^2h_H^2+r_I^2h_H^2-d_{I,1}^2h_I^2+h_H^2h_I^2}{2r_I r_Uh_H^2} \big)$.
Therefore, $S_I(d_{I,1})$ can be calculated using: 
\begin{equation}
\begin{array}{r@{}l}
S&_I(d_{I,1})=S_3(d_{I,1})-S_2(d_{I,1})=\phi(d_{I,1})(d_{I,1}^2-h_H^2)\\
    &-\sin\eta(d_{I,1})r_Ir_U\frac{h_H^2}{h_I^2}+\eta(d_{I,1})r_I^2\frac{h_H^2}{h_I^2}.
\end{array}
\end{equation}

The CDF can be formulated as:
\begin{equation}
\begin{array}{r@{}l}
    F&(d_{I,1})=1-\exp(\lambda_{H}S_I(d_{I,1}))\\
    &=1-\exp(-\lambda_{H}\phi(d_{I,1})(d_{I,1}^2-h_H^2))\\
    &\times\exp(\lambda_{H}\sin\eta(d_{I,1})r_Ir_U\frac{h_H^2}{h_I^2}-\lambda_{H}\eta(d_{I,1})r_I^2\frac{h_H^2}{h_I^2}),
\end{array}
\end{equation}
and the PDF of $d_{I,1}$ is:
\begin{equation}
\begin{array}{r@{}l}
    f&(d_{I,1})=2\lambda_{H}\phi(d_{I,1})d_{I,1}\exp(-\lambda_{H}\phi(d_{I,1})(d_{I,1}^2-h_H^2))\\
    &\times\exp(\lambda_{H}\sin\eta(d_{I,1})r_Ir_U\frac{h_H^2}{h_I^2}-\lambda_{H}\eta(d_{I,1})r_I^2\frac{h_H^2}{h_I^2}).
\end{array}
\end{equation}

For the case where $d_{I,1}>d_{th,2}$, the area can be divided into two parts, the first part is the intersection of the circular area centered at $\textbf{x}_{U}'$ with radius $\sqrt{d_{th,2}^2-h_H^2}$, and the circular area centered at $\textbf{o}_{I}'$ with radius $r_I\frac{h_H}{h_I}$. We have
    $S_3(d_{I,1})=\theta(d_{th,2}^2-h_H^2)$,
and 
    $S_2(d_{I,1})=\cot\theta r_I^2\frac{h_H^2}{h_I^2}-(\frac{\pi}{2}-\theta)r_I^2\frac{h_H^2}{h_I^2}$,
so that the area of the first part is $S_3(d_{I,1})-S_2(d_{I,1})$. The second part is an annular sector centered at $\textbf{z}_U'$ with inner radius $d_{th,2}$, outer radius $d_{I,1}$. Therefore the areas of the second part is:
    $S_4(d_{I,1})=\theta(d_{I,1}^2-d_{th,2}^2)$.
Therefore, $S_{I}(d_{I,1})$ can be calculated using
\begin{equation}
\begin{array}{r@{}l}
    S&_{I}(d_{I,1})=S_{3}(d_{I,1})+S_{2}(d_{I,1})+S_{4}(d_{I,1})\\
    &=\theta(d_{I,1}^2-h_H^2)+(\frac{\pi}{2}-\theta)r_I^2\frac{h_H^2}{h_I^2}-\cot\theta r_I^2\frac{h_H^2}{h_I^2}.
\end{array}
\end{equation}

Therefore, the CDF can be calculated using
\begin{equation}
\begin{array}{r@{}l}
    F&(d_{I,1})=1-\exp(-\lambda_{H}S_I(d_{I,1}))\\

    &=1-\exp(-\lambda_{H}\theta(d_{I,1}^2-h_H^2))\\
    &\times\exp(\lambda_{H}\cot\theta r_I^2\frac{h_H^2}{h_I^2}-\lambda_{H}(\frac{\pi}{2}-\theta)r_I^2\frac{h_H^2}{h_I^2}),
\end{array}
\end{equation}
and the PDF of $d_{I,1}$ is
\begin{equation}
\begin{array}{r@{}l}
    f&(d_{I,1})=2\lambda_{H}\theta d_{I,1}\exp(-\lambda_{H}\theta(d_{I,1}^2-h_H^2))\\
    &\times\exp(\lambda_{H}\cot\theta r_I^2\frac{h_H^2}{h_I^2}-\lambda_{H}(\frac{\pi}{2}-\theta)r_I^2\frac{h_H^2}{h_I^2}).
\end{array}
\end{equation}
We can summarize the PDF of $d_{I,1}$ as:
\begin{equation}
\begin{array}{l@{}l}
&f(d_{I,1})=2\lambda_{H}d_{I,1}(1-F(d_{I,1}))\\
&\times (\phi(d_{I,1})\mathbbm{1}(d_{th,1}<d_{I,1}<d_{th,2})+\theta\mathbbm{1}(d_{I,1}>d_{th,2})).
\end{array}
\end{equation}
\section{Proof of Theorem \ref{theo:accurate_for_SA}}\label{app:accurate_for_SA}
As shown in (\ref{PcovSum}), we can calculate the coverage probability using the sum of two terms:
\begin{equation}
    P_{\rm cov}(r_U,\tau)=P_{\rm cov}(r_U,\tau,\textbf{H}_O)+P_{\rm cov}(r_U,\tau,\textbf{H}_I).
\end{equation}

Similar to the offshore case, the coverage probability conditioned with TX selection case $\textbf{H}_O$ can be calculated using:
\begin{equation}
\begin{array}{r@{}l}
  &P_{\rm cov}(r_U,\tau|\textbf{H}_O,d_{O,1},d_{I,1})\\
  &=P_{\rm cov}(r_U,\tau|d_{I,1}> \mathcal{T}_{I}(d_{O,1}),d_{O,1})\\
  &= \P \big\{\frac{p_{O,0}W_{O,1}d_{O,1}^{-\alpha_O}}{N_0+I(d_{O,1},d_{I,1})}>\tau \big\}
  =\overline{F}_{W_{O,1}}\big(g_{O}^{h}(d_{O,1},d_{I,1})\big),
\end{array}
\end{equation}
where 
    $g_{O}^{h}(d_{O,1},d_{I,1})= \frac{\tau}{p_{O,0}}d_{O,1}^{\alpha_{O}} \big[N_{0}+I(d_{O,1},d_{I,1}) \big]$.
We define
    $\mathcal{T}_d(t)=\frac{h_H}{h_I}(r_U\cos{t}-\sqrt{r_I^2-r_U^2\sin^2{t}})$,
so that the interference from HAPS based on $d_{O,1}$ and $d_{I,1}$ is
\begin{equation}
\begin{array}{r@{}l}
I&(d_{O,1},d_{I,1})=I_O(d_{O,1})+I_I(d_{I,1})\\
&\approx \mathbb{E}(\sum\limits_{d_{O,j}>d_{O,1}} p_{O,0}W_{O,j}d_{O,j}^{-\alpha_O}+\sum\limits_{d_{I,k}>d_{I,1}} p_{I,0}W_{I,k}d_{I,k}^{-\alpha_I})\\
& =p_{O,0}(2b_{O}+\Omega_{O})\frac{(2\pi-2\theta)\lambda_{H}}{\alpha_O-2}z^{2-\alpha_O} \big|_{\infty}^{d_{O,1}}\\
& +p_{O,0}(2b_{O}+\Omega_{O})\frac{\lambda_{H}}{\alpha_O-2}\int_{-\theta}^{\theta}z^{2-\alpha_O} \big|_{\mathcal{T}_d(t)}^{\min\{d_{O,1},\mathcal{T}_d(t)\}} \dd t\\
& +p_{I,0}(2b_{I}+\Omega_{I})\frac{\lambda_{H}}{\alpha_I-2}\int_{-\theta}^{\theta}z^{2-\alpha_I} \big|_{\infty}^{\max\{d_{I,1},\mathcal{T}_d(t)\}} \dd t.

\end{array}
\label{IdO1dI1}
\end{equation}

Therefore, the coverage probability with TX selection case $\textbf{H}_O$ can be calculated using:
\begin{equation}
\begin{array}{r@{}l}
    P&_{\rm cov}(r_U,\tau,\textbf{H}_O)\\
    &= \int_{h_H}^{\infty} \int_{d_{th,1}}^{\infty} P_{\rm cov}(r_U,\tau|\textbf{H}_O,d_{O,1},d_{I,1})\\
    &\;\;\;\;\;\;\;\;\;\;\;\;\;\;\;\;\;\;\;\;\;\;\;\;\times f(d_{O,1})f(d_{I,1})\;\dd d_{I,1}\dd d_{O,1}\\
    &= \int_{h_H}^{\infty} \int_{\max\{d_{th,1},\mathcal{T}_{I}(d_{O,1})\}}^{\infty} P_{\rm cov}(r_U,\tau|d_{O,1},d_{I,1})\\
    &\;\;\;\;\;\;\;\;\;\;\;\;\;\;\;\;\;\;\;\;\;\;\;\;\times f(d_{O,1})f(d_{I,1})\;\dd d_{I,1}\dd d_{O,1}\\
    &= \int_{h_H}^{\infty} \int_{\max\{d_{th,1},\mathcal{T}_{I}(d_{O,1})\}}^{\infty} \overline{F}_{W_{O,1}}\big(g_{O}^{h}(d_{O,1},d_{I,1})\big)\\
    &\;\;\;\;\;\;\;\;\;\;\;\;\;\;\;\;\;\;\;\;\;\;\;\;\times f(d_{O,1})f(d_{I,1})\;\dd d_{I,1}\dd d_{O,1}.
\end{array}
\label{PcovRuTauHO}
\end{equation}

Similar to the onshore case, the coverage probability with TX 
selection case $\textbf{H}_{I}$ can be calculated using:
\begin{equation}
\begin{array}{r@{}l}
  &P_{\rm cov}(r_U,\tau|\textbf{H}_I,d_{O,1},d_{I,1})\\
  &=P_{\rm cov}(r_U,\tau|d_{O,1}> \mathcal{T}_{O}(d_{I,1}),d_{I,1})\\
  &= \P \big\{\frac{p_{I,0}W_{I,1}d_{I,1}^{-\alpha_I}}{N_0+I(d_{O,1},d_{I,1})}>\tau \big\}
  =\overline{F}_{W_{I,1}}\big(g_{I}^{h}(d_{O,1},d_{I,1})\big),
\end{array}
\end{equation}
where
    $g_{I}^{h}(d_{O,1},d_{I,1})= \frac{\tau}{p_{I,0}}d_{I,1}^{\alpha_{I}} \big[N_{0}+I(d_{O,1},d_{I,1}) \big]$.
The interference is shown in (\ref{IdO1dI1}) and the coverage probability with TX selection case $\textbf{H}_I$ can be calculated using:
\begin{equation}
\begin{array}{r@{}l}
    P&_{\rm cov}(r_U,\tau,\textbf{H}_I)\\
    &= \int_{d_{th,1}}^{\infty}\int_{h_H}^{\infty}  P_{\rm cov}(r_U,\tau|\textbf{H}_I,d_{O,1},d_{I,1})\\
    &\;\;\;\;\;\;\;\;\;\;\;\;\;\;\;\;\;\;\;\;\;\;\;\;\times f(d_{I,1})f(d_{O,1})\;\dd d_{O,1}\dd d_{I,1}\\
    &= \int_{d_{th,1}}^{\infty}\int_{\max\{h_H,\mathcal{T}_{O}(d_{I,1})\}}^{\infty} P_{\rm cov}(r_U,\tau|d_{O,1},d_{I,1})\\
    &\;\;\;\;\;\;\;\;\;\;\;\;\;\;\;\;\;\;\;\;\;\;\;\;\times f(d_{O,1})f(d_{I,1})\;\dd d_{I,1}\dd d_{O,1}\\
    &= \int_{d_{th,1}}^{\infty}\int_{\max\{h_H,\mathcal{T}_{O}(d_{I,1})\}}^{\infty} \overline{F}_{W_{I,1}}\big(g_{I}^{h}(d_{O,1},d_{I,1})\big)\\
    &\;\;\;\;\;\;\;\;\;\;\;\;\;\;\;\;\;\;\;\;\;\;\;\;\times f(d_{O,1})f(d_{I,1})\;\dd d_{O,1}\dd d_{I,1}.
\end{array}
\label{PcovRuTauHI}
\end{equation}

\section{Proof of Lemma \ref{lem:CDFPDFofdO1dI1}}\label{app:CDFPDFofdO1dI1}
Based on the approximation method 1, we model the HAPSs operating in the Rician channels and the HAPSs operating in the shadowed Rician channels using two independent PPPs with density $\lambda_H$, respectively. We discuss the cases where $h_H<d_{I,1}<d_{th,1}$ and $d_{th,1}<d_{I,1}<\infty$. When $h_H<d_{I,1}<d_{th,1}$, the area of the Rician region within the distance $d_{th,1}$ to $\textbf{z}_U$ is
    $S_O'(d_{O,1})=\pi (d_{O,1}^2-h_H^2)$,
and when $d_{th,1}<d_{I,1}<\infty$ the area of the Rician region contains a sector area centered at $\textbf{z}_U'$, with radius $d_{O,1}$ and central angle $2\pi-2\theta$, and a sector area centered at $\textbf{z}_U'$ with radius $d_{th,1}$. Therefore the area of the Rician region is
    $S_O'(d_{O,1})=(\pi-\theta)(d_{O,1}^2-h_H^2)+\theta(d_{th,1}^2-h_H^2)$.
In short, we have
\begin{equation}
    S_O'(d_{O,1})=(\pi-\theta)d_{O,1}^2+\theta \min\{d_{O,1},d_{th,1}\}^2-\pi h_H^2.
\end{equation}

Therefore, the CDF of the distance between the typical user and its closest HAPS in Rician region $d_{O,1}$ is:
\begin{equation}
\begin{array}{r@{}l}
    F&(d_{O,1})=1-\exp(-\lambda_{H}S_O'(d_{O,1}))\\
    &=1-\exp(-\lambda_H (\pi-\theta)d_{O,1}^2-\lambda_{H}\theta\min(d_{O,1},d_{th,1})^2)\\
    &\times \exp(\lambda_{H}\pi h_H^2),
\end{array}
\end{equation}
and the PDF of $d_{O,1}$ is the derivative of $F(d_{O,1})$:
\begin{equation}
\begin{array}{r@{}l}
f&(d_{O,1})=2\lambda_Hd_{O,1} \big((\pi-\theta)+\theta \mathbbm{1}(d<d_{th,1}) \big)\\
&\times\exp(-\lambda_H (\pi-\theta)d_{O,1}^2-\lambda_{H}\theta\min(d_{O,1},d_{th,1})^2)\\
&\times \exp(\lambda_{H}\pi h_H^2),
\end{array}
\end{equation}
for $d_{O,1}>h_H$.

Similarly, the Rician region within the distance $d_{I,1}$ is an annular sector area centered at $\textbf{z}_U'$, with inner radius $r_1$, infinite outer radius and central angle $2\theta$. Therefore, the area of the Rician region within the distance $d_{I,1}$ is:
\begin{equation}
    S_I'(d_{I,1})=\theta(d_{I,1}^2-d_{th,1}^2).
\end{equation}

Therefore, the CDF of $d_{I,1}$ can be written as:
\begin{equation}
\begin{array}{r@{}l}
F(d_{I,1})&=1-\exp(-\lambda_{H}S_I'(d_{I,1}))\\
&=1-\exp(-\lambda_{H}\theta(d_{I,1}^2-d_{th,1}^2)),

\end{array}
\end{equation}
and the PDF of $d_{I,1}$ is the derivative of $F(d_{I,1})$:
\begin{equation}
f(d_{I,1})=2\lambda_{H}\theta d_{I,1}\exp(-\lambda_{H}\theta(d_{I,1}^2-d_{th,1}^2)),
\end{equation}
for $d_{I,1}>d_{th,1}$.

\section{Proof of Theorem \ref{theo:approx1_for_SA}}\label{app:approx1_for_SA}
As shown in (\ref{PcovSum}), we can calculate the coverage probability using the sum of two terms:
\begin{equation}
    P_{\rm cov}(r_U,\tau)=P_{\rm cov}(r_U,\tau,\textbf{H}_O)+P_{\rm cov}(r_U,\tau,\textbf{H}_I).
\end{equation}

The same as the proof of Theorem \ref{theo:accurate_for_SA}, we can also obtain the coverage probability of nearshore users with $\textbf{H}_O$ is
\begin{equation}
\begin{array}{r@{}l}
    P&_{\rm cov}(r_U,\tau,\textbf{H}_O)= \int_{h_H}^{\infty} \int_{\max\{d_{th,1},\mathcal{T}_{I}(d_{O,1})\}}^{\infty}\\ &\overline{F}_{W_{O,1}}\big(g_{O}^{h}(d_{O,1},d_{I,1})\big) f(d_{O,1})f(d_{I,1})\;\dd d_{I,1}\dd d_{O,1},
\end{array}
\end{equation}
with
    $g_{O}^{h}(d_{O,1},d_{I,1})= \frac{\tau}{p_{O,0}}d_{O,1}^{\alpha_{O}} \big[N_{0}+I(d_{O,1},d_{I,1}) \big],$
and the coverage probability of nearshore users with $\textbf{H}_I$ is

\begin{equation}
\begin{array}{r@{}l}
    P&_{\rm cov}(r_U,\tau,\textbf{H}_I)= \int_{d_{th,1}}^{\infty}\int_{\max\{h_H,\mathcal{T}_{O}(d_{I,1})\}}^{\infty}\\ &\overline{F}_{W_{I,1}}\big(g_I^{h}(d_{O,1},d_{I,1})\big) f(d_{O,1})f(d_{I,1})\;\dd d_{O,1}\dd d_{I,1},
\end{array}
\end{equation}
with
    $g_{I}^{h}(d_{O,1},d_{I,1})= \frac{\tau}{p_{I,0}}d_{I,1}^{\alpha_{I}} \big[N_{0}+I(d_{O,1},d_{I,1}) \big]$.
Differently, based on the approximation method 1, we have approximate CDFs and PDFs of $d_{O,1}$ and $d_{I,1}$ formulated in Lemma \ref{lem:CDFPDFofdO1dI1}. In this case, interference conditioned on $d_{O,1}$ and $d_{I,1}$ can be calculated using it expectation
\begin{equation}
\begin{array}{r@{}l}
    I&(d_{O,1},d_{I,1})=I_O(d_{O,1})+I_I(d_{I,1})\\
    &\approx \mathbb{E}(\sum\limits_{d_{O,j}>d_{O,1}} p_{O,0}W_{O,j}d_{O,j}^{-\alpha_O}+\sum\limits_{d_{I,k}>d_{I,1}} p_{I,0}W_{I,k}d_{I,k}^{-\alpha_I})\\

    & =p_{O,0}(2b_{O}+\Omega_{O})\frac{(2\pi-2\theta)\lambda_{H}}{\alpha_{O}-2}z^{2-\alpha_{O}} \big|_{\infty}^{d_{O,1}}\\
    & \;\;\;+p_{O,0}(2b_{O}+\Omega_{O})\frac{2\theta\lambda_{H}}{\alpha_{O}-2}z^{2-\alpha_O} \big|_{d_{th,1}}^{\min\{d_{O,1},d_{th,1}\}}\\
    & \;\;\;+p_{I,0}(2b_{I}+\Omega_{I})\frac{2\theta\lambda_{H}}{\alpha_{I}-2}z^{2-\alpha_{I}} \big|_{\infty}^{d_{I,1}},
\end{array}
\end{equation}
which has a lower computation complexity than the interference for exact expression of coverage probability.

\ifCLASSOPTIONcaptionsoff
  \newpage
\fi

\bibliographystyle{IEEEtran}
\bibliography{ref}

\end{document}